\documentclass[prodmode,acmjacm]{acmsmall}

\usepackage[ruled]{algorithm2e}
\usepackage{amssymb}
\usepackage{epsfig}

\SetAlFnt{\small}
\SetAlCapFnt{\small}
\SetAlCapNameFnt{\small}
\SetAlCapHSkip{0pt}
\IncMargin{-\parindent}

\acmVolume{0}
\acmNumber{0}
\acmArticle{XX}
\acmYear{0000}
\acmMonth{1}

\SetKwFunction{getTS}{getTS}
\SetKwFunction{compare}{compare}
\SetKwFunction{scan}{scan}
\SetKwFunction{simpgetTS}{simple-getTS}
\SetKwFunction{simpcompare}{simple-compare}

\newcommand{\union}{\cup}
\newcommand{\ID}{\ensuremath{\mathrm{ID}}}
\newcommand{\sig}[1]{\ensuremath{\mathrm{sig}(#1)}}
\newcommand{\tre}[2]{\mathcal{R}_{#1}(#2)}
\newcommand{\ordsig}[1]{\ensuremath{\mathrm{ordSig}(#1)}}
\newcommand{\procs}[1]{\ensuremath{\mathrm{participants}(#1)}}

\newcommand{\allRegs}{\ensuremath {\mathcal{R}}}
\newcommand{\full}[2]{$(#1,#2)${-full}}
\newcommand{\fullRegs}[3]{\ensuremath {\mathcal{R}_{#1,#2}(#3)}}

\newcommand{\poised}[2]{\ensuremath{\mathrm{poised}(#1,#2)}} 
\newcommand{\idle}[1]{\ensuremath{\mathrm{idle}(#1)}} 
\newcommand{\reg}{\ensuremath{{r}}} 

\newtheorem{claim}[theorem]{Claim}

\begin{document}
\sloppy
\markboth{Helmi et al.}{The Space Complexity of Long-lived and One-Shot Timestamp Implementations}

\title{The Space Complexity of Long-lived and One-Shot Timestamp Implementations}
\author{Maryam Helmi
\affil{University of Calgary}
Lisa Higham
\affil{University of Calgary}
Eduardo Pacheco
\affil{Universidad Nacional Automoma de Mexico}
Philipp Woelfel
\affil{University of Calgary}
}
\begin{abstract}
This paper is concerned with the problem of implementing an unbounded timestamp
  object from multi-writer atomic registers, in an asynchronous distributed
  system of $n$ processes with distinct identifiers where timestamps are taken from
  an arbitrary universe.  Ellen, Fatourou
  and Ruppert~\cite{EFR2008a} showed that $\sqrt{n}/2-O(1)$ registers are
  required for any obstruction-free implementation of long-lived 
  timestamp systems from atomic registers  
  (meaning processes can repeatedly get timestamps). 

  We improve this existing lower bound in two ways.  
  First we establish a lower bound of $n/6 - 1$ registers 
  for the obstruction-free long-lived timestamp problem.
  Previous such linear lower bounds were only known for constrained versions of the timestamp problem. 
  This bound is asymptotically tight; 
  Ellen, Fatourou and Ruppert~\cite{EFR2008a}  constructed a wait-free algorithm that uses $n-1$ registers. 
  Second we show that $\sqrt{2n} - \log n - O(1)$ registers are required 
  for any obstruction-free implementation of one-shot timestamp systems
(meaning each process can get a timestamp at most once).
We show that this bound is also asymptotically tight by providing a wait-free one-shot timestamp system that uses  
at most $\lceil 2 \sqrt{n} \rceil$ registers, 
thus establishing a space complexity gap between one-shot and long-lived  timestamp systems.  
\end{abstract}




\begin{bottomstuff}
This work is supported by the National Sciences and Research Council of Canada Discovery Grants. 
E. Pacheco participated in this research while visiting the University of Calgary.

Author's addresses: M. Helmi {and} L. Higham {and} P. Woelfel, Computer Science Department,
University of Calgary;  
E. Pacheco,
Computer Science Department, University of Ottawa, Canada. 
\end{bottomstuff}

\maketitle

\section{Introduction}

In asynchronous multiprocessor algorithms, processes have no information about the real-time order of events 
that are incurred by other processes.
In order to solve distributed problems effectively, such as ensuring first-come-first-served fairness, or constructing synchronization primitives, it is often necessary that some reliable information about the relative order of these events can be gained.

Timestamp objects provide a means for processes to label events and then later compare those labels in order to gain information about the real-time order in which the corresponding events have occurred.
Such timestamping mechanisms have been used to solve numerous problems associated with asynchrony in distributed shared memory and message passing algorithms.
Examples of applications 
include mutual and $k$-exclusion algorithms \cite{Lam1974a,RA1981a,FLBB1989a,ADGMS1994a}, consensus algorithms \cite{Abr1988a},  
register constructions 
\cite{HV2002a,LTV1996a,VA1986a}, or adaptive renaming algorithms \cite{AF2003a}.

In 1978, Lamport \cite{Lam1978a} defined the ``happens before'' relation on events occurring in message passing systems to reflect the causal relationship of events.
The happens before relation is a partial order, where, informally, an event $e_1$ happens before event $e_2$, if there is a causal relation that forces 
event $e_1$ to precede $e_2$.
Lamport further devised a \emph{logical clock} that assigns an integer value $C(e)$, called a timestamp, to each event $e$ 
such that $C(e_1)<C(e_2)$ if event $e_1$ happens before event $e_2$.
Lamport's logical clock system based on integers was extended to clocks based on vectors 
(examples include \cite{Fid1988a} and \cite{Mat1989a}) and matrices (\cite{WB1986a} and \cite{SL1987a}).

In shared memory systems, 
events correspond to method invocations and responses. 
The happens before relation orders time intervals associated with method calls.
Method call $m_1$ happens before method call $m_2$, if the response of $m_1$ precedes the invocation of $m_2$.
Timestamp objects provide a mechanism to label events with timestamps from a \emph{timestamp universe} $\mathcal{T}$ through \getTS{} 
(sometimes called \FuncSty{timestamping} or \FuncSty{label}) method calls.
If $\mathcal{T}$ is finite, then the timestamp object is said to be \emph{bounded}, otherwise it is \emph{unbounded}.

Often, $\mathcal{T}$ is a partially ordered set, and all timestamps returned by \getTS{} method calls during an execution preserve the happens before relation of these method calls.
Such timestamp objects are called \emph{static}.
Non-static timestamp objects can take the current system state into account when comparing the order of two timestamps.
Thus, different executions can lead to different partial orders of the set $\mathcal{T}$.
Sometimes, in particular when $\mathcal{T}$ is bounded, the happens before relation is only preserved for a subset of \emph{valid} timestamps in $\mathcal{T}$, e.g., the set of the last timestamps obtained by each process.
In this case, timestamp objects often provide a \FuncSty{scan} method that returns an ordered list of all valid timestamps.
The literature contains several examples of constructions of bounded and unbounded timestamp objects \cite{Lam1974a,GLS1992a,IP1992a,IL1993a,DS1997a,DW1999a,HV2002a,AF2003a,GR2007a,EFR2008a}.

Ellen, Fatourou, and Rupert \cite{EFR2008a} studied the number of atomic registers needed to implement  timestamp objects.
In order to prove strong lower bounds, the authors considered a very weak definition of an unbounded non-static timestamp object, 
that, in addition to \getTS{} provides a method \compare{$t_1,t_2$} for two timestamps $t_1,t_2\in\mathcal{T}$.
The only requirement is that if a \getTS{} method $g_1$ that returns $t_1$
happens before another \getTS{} method $g_2$ that returns $t_2$ then
any later \compare{$t_1,t_2$} must return true and any later \compare{$t_2,t_1$} must return false.

As their main result, Ellen et al.\ showed that any implementation that satisfies non-deterministic solo-termination (a progress condition weaker than wait-freedom or obstruction-freedom, and that is defined in Section~\ref{section:preliminaries})
requires at least $\frac12\sqrt{n-1}$ registers, where $n$ is the number of processes in the system.
Despite the weak requirements, the best known algorithm (also in \cite{EFR2008a}) needs $n-1$ registers, leaving a large gap between the best known lower and upper bounds.
However, for two stronger versions of the problem, Ellen et al.\ obtain tight lower bounds, showing that $n$ registers are necessary,
first, for static algorithms, where $\mathcal{T}$ is \emph{nowhere dense} (i.e., any two elements $x,y\in\mathcal{T}$ 
satisfy $|\{z\in\mathcal{T}\,|x<z<y\}|<\infty$), 
and second, for anonymous algorithms.

\subsubsection*{Our Contributions}
We distinguish between \emph{one-shot} timestamp objects, where each process is allowed to call \getTS{} at most once, and \emph{long-lived} ones, where each process can call \getTS{} arbitrarily many times.
(In either case, the number of \compare methods calls is not restricted.)
We first improve the $\Omega(\sqrt n)$ lower bound of \cite{EFR2008a} for long-lived timestamp objects to an asymptotically tight one:
\begin{theorem}\label{theorem:long-lived}\samepage
  Any long-lived unbounded timestamp object that satisfies non-deterministic solo-termination uses at least $n/6-1$ registers.
\end{theorem}
Therefore, even under very weak assumptions, at least linear register space is necessary.
Since it is not possible to implement general timestamp objects using sublinear space, it makes sense to look at restricted solutions.

Several methods have solutions that are simpler than the general case, if each process is allowed to execute it only once.
Examples are renaming and mutual exclusion algorithms, splitter or snapshot objects, or agreement problems.
Other problems, such as consensus or non-resettable test and set objects are inherently ``one-time''.
It is conceivable that if an implementation of such an algorithm uses timestamp objects, 
then in the ``one-shot'' version of that algorithm each process needs to obtain a timestamp only once.
Therefore, we study the space complexity of one-shot timestamp objects:
\begin{theorem}
\label{theorem:one-timeLbd}
\samepage
  Any one-shot unbounded timestamp object that satisfies non-deterministic solo-termination uses at least 
$\sqrt{2n}-\log n -O(1)$ registers.
\end{theorem}

This one-shot lower bound is a factor of approximately $2\sqrt{2}$ larger than the previous best known lower bound 
for the long-lived case \cite{EFR2008a}, 
and holds for historyless objects as well as registers as explained later.

\begin{theorem}
\label{theorem:one-timeUbd}
\samepage
There is a  wait-free implementation of one-shot timestamp objects that uses $2 \lceil \sqrt{n} \rceil$ registers.
\end{theorem}

Our lower bound proofs are based on covering arguments (as introduced by Burns and Lynch \cite{BurnsL93}), 
where one constructs an execution in which processes are poised to write to some registers 
(the processes are said to \emph{cover} these registers). 
We rely on a lemma by Ellen, Fatourou and Ruppert \cite{EFR2008a} that shows how in a situation where some processes cover a set $R$ of registers, other processes can be forced to write outside of $R$.
In order to obtain our improved lower bound for the long-lived case, we look at very long executions in which ``similar'' coverings are obtained over and over again.  
Our lower bound proof for the one-shot case is inspired by a geometric interpretation of the covering structure of configurations. 
The one-shot timestamps upper bound exploits the structure exposed by the lower bound.

\section{Preliminaries}\label{section:preliminaries}

 We consider an asynchronous shared memory system with a set $\mathcal{P}$ $= \{p_1,\dots,p_n\}$ 
of $n$ processes and a set $\mathcal{R} = \{ r_1,\dots,r_m \}$ of $m$ registers that support atomic read and write operations.
Processes can only communicate via those operations on shared registers.
We assume that processes can make arbitrary non-deterministic decisions, but we require that the result of any  execution is correct, 
meaning that the responses from method calls match the specification of timestamp objects.

A \emph{configuration} $C$ is a tuple $(s_1,\dots,s_n,v_1,\dots,v_m)$, denoting that process $p_i$, $1\leq i\leq n$, is in state $s_i$, and register $r_j$, $1\leq j\leq m$, has value $v_j$.
Configurations will be denoted by capital letters, and the initial configuration is denoted  $C_0$.

An implementation of a method satisfies \emph{non-deterministic solo-termination}, if for any configuration $C$ and any process $p_i$, $1\leq i\leq n$, there is an execution in which no process other than $p_i$ takes any steps, and $p_i$ finishes its method call within a finite number of steps \cite{FHS1998a}.
Hence, a process is guaranteed to finish its method call with positive probability, whenever there is no interference from other processes.
For deterministic algorithms, non-deterministic solo-termination is the same as obstruction-freedom and weaker than wait-freedom.
Both our lower bound results hold for timestamp objects that satisfy this progress condition, 
our algorithm, however, satisfies the stronger wait-free progress property.

A \emph{schedule} $\sigma$ is a (possibly infinite) sequence of process indices.
An \emph{execution} $(C;\sigma)$ is a sequence of steps beginning in configuration $C$ and moving through successive configurations one at a time. 
At each step, the next process $p_i$ indicated in the schedule $\sigma$, takes the next step in its program.
Since our computation model is non-deterministic, we fix the non-deterministic decision made by $p_i$ in our lower bound proofs.
We use an arbitrary (but fixed) one that guarantees that $p_i$ terminates within a bounded number of steps if it executes alone.
If $\sigma$ is a finite schedule, the final configuration of the execution $(C;\sigma)$ is denoted $\sigma(C)$.
If $\sigma$ and $\pi$ are finite schedules then  $\sigma\pi$ denotes the concatenation of $\sigma$ and $\pi$.
Let $P$ be a set of processes,  and $\sigma$ a schedule. 
If only indices of processes in $P$ appear in $\sigma$, 
then $\sigma$ is a \emph{$P$-only} schedule and any execution $(C; \sigma)$ is a \emph{$P$-only} execution. 
If $|P| =1$, a $P$-only  schedule $\sigma$ is a \emph{solo} schedule and any execution $(C; \sigma)$ is a \emph{solo} execution.

A configuration, $C$, is \emph{reachable} if there exists a finite schedule, $\sigma$, such that $\sigma(C_0)= C$. 

Any execution $(C;\sigma)$ defines a partial \emph{happens before} order ``$\rightarrow$'' on the method calls that occur during $(C;\sigma)$.
A method call $m_1$ happens before $m_2$, denoted $m_1\rightarrow m_2$, if the response of $m_1$ occurs before the invocation of $m_2$.

An unbounded timestamp object supports two methods, \getTS{} and \compare{}. 
The first one outputs a \emph{timestamp} without receiving any input; 
the \compare method receives any two timestamps as inputs, and returns true or false.
If  two \getTS{} instances $g_1$ and $g_2$ return $t_1$ and $t_2$, respectively, 
and $g_1\rightarrow g_2$, then \compare{$t_1,t_2$} returns true and \compare{$t_2$,$t_1$} returns false.

A timestamp object is \emph{long-lived}, 
if each process is allowed to invoke \getTS{} multiple times; 
it is \emph{one-shot} when each process is allowed to invoke \getTS{} only once.

Our lower bounds are  based on covering arguments.
We will construct executions, at the end of which processes are poised to write, i.e., they \emph{cover} several registers.
If other process are scheduled after this and if they write only to the same set of registers, their trace can be eliminated.
More precisely,
we say process $p_i$ \emph{covers} register $r_j$ in a configuration $C$, 
if there is a non-deterministic decision 
such that the one step execution $\bigl(C;(i)\bigr)$ is a write to register $r_j$. 
A set of processes $P$ covers a set of registers $R$ if for every register $r\in R$ there is a process $p\in P$ 
such that $p$ covers $r$. 

For a process set $P$, let $\pi_{P}$ denote an arbitrary (but fixed) permutation of $P$ 
(for example the one that orders processes by their ID).
If the process set $P$ covers the register set $R$ in configuration $C$, 
the information held in the registers in $R$ can be overwritten by letting all processes in $P$ execute exactly one step.
Such an execution by the processes in $P$ is called a \emph{block-write}.
More precisely, 
a \emph{block-write} by $P$ to $R$ is an execution $(C;\pi_{P})$. 

Two configurations $C_1=(s_1,\ldots,s_n,r_1,\ldots,r_m)$ and 
$C_2=(s'_1,\ldots,s'_n,r'_1,\ldots,r'_m)$ are \emph{indistinguishable} to process $p_i$ 
if $s_i=s'_i$ and $r_j=r'_j$ for $1\leq j\leq n$. 
If $S$ is a set of processes, and for every process $p\in S$, $C_1$ and $C_2$ are indistinguishable to $p$, 
then for any $S$-only schedule $\sigma$,  $\sigma(C_1)$ and $\sigma(C_2)$ are indistinguishable to $p$. 

Our first lower lower bound relies on a lemma 
which is based on the following observation. 
Suppose in configuration $C$ there are three disjoint sets of processes $B_0,B_1,B_2$, 
each covering a set ${R}$ of registers, 
and $q_0$ and $q_1$ are processes not in $B_0\cup B_1\cup B_2$.
Let $\sigma_i$, $i \in \{0,1\} $, denote an arbitrarily long $\{q_i\}$-only schedule.
If, for  $i \in \{0,1\} $, in the execution $(C;\pi_{B_i}\sigma_i)$, $q_i$ does not write outside $R$, 
then the configurations $\pi_{B_i}\sigma_i(C)$  and $\pi_{B_{i-1}}\sigma_{i-1}\pi_{B_i}\sigma_i(C)$ 
are indistinguishable to $q_i$. 
Furthermore, after a subsequent third block write by $B_2$ all trace left inside of $R$ can also be obliterated.
Thus, the configurations $C_0=\pi_{B_0} \sigma_0\pi_{B_1} \sigma_1\pi_{B_2}(C)$ 
and $C_1=\pi_{B_1} \sigma_1\pi_{B_0} \sigma_0\pi_{B_2}(C)$ are indistinguishable to all processes, 
unless at least one of either $q_0$ or $q_1$ writes outside ${R}$.
If, however, the solo executions by $q_0$ and $q_1$ both contain complete \getTS{} calls, 
then one happens after the other and so processes have to be able to distinguish between $C_0$ and $C_1$.
Hence, either $q_0$ or $q_1$ writes outside ${R}$ in both of the executions
$(C;\pi_{B_{0}}\sigma_{0} \pi_{B_{1}} \sigma_{1})$ and 
$(C;\pi_{B_{1}}\sigma_{1} \pi_{B_{0}} \sigma_{0})$.

The same idea works if we replace $q_0$ and $q_1$ with disjoint sets of processes,
as was done in the original version of this lemma due to 
Ellen, Fatourou, and Rupert \cite{EFR2008a}.
We state a simplified form here that suffices for our results and uses the 
form and notation of this paper.

\begin{lemma}[\cite{EFR2008a}]
\label{lemma:EFR2008a-main}
  Consider any timestamp implementation from registers that satisfies non-deterministic solo-termination 
  and let $C$ be a reachable configuration. 
  Let $B_0$, $B_1$, $B_2, U_0, U_1$ be disjoint sets of processes, where in $C$ each of $B_0$, $B_1$, and $B_2$ cover a set ${R}$ of registers.
  Then there exists $i\in\{0,1\}$ such that every $U_i$-only execution starting from $C_i=\pi_{B_i}(C)$ 
  that contains a complete \getTS{} method writes to some register not in ${R}$.
\end{lemma}

Our second lower bound relies on a stronger lemma that is proved by inductively applying Lemma \ref{lemma:EFR2008a-main}.

%

\section{A Space Lower Bound for Long-Lived Timestamps}
\label{section:long-lived-lwb}

We assume that a timestamp object is used in an algorithm where each process calls \getTS{} infinitely many times.
Actually, the number of \getTS{} calls can be bounded (by a function growing exponentially in $n$), 
but for convenience we pass on computing this bound.
Ellen et al.\ used their lemma in order to inductively construct executions at the end of which $k$ registers are covered by $\Omega(\sqrt n-k)$ processes, where $k$ is bounded by $O(\sqrt n)$.
The lemma is used in the inductive step to show that in some execution following a block-write, many of the non-covering processes can be forced to write outside the set of covered registers.
By the pigeon hole principle, one additional (previously not covered) register can then be covered with many processes.
With this idea, however, the number of processes covering one register is reduced by one in each inductive step, and thus it is not hard to see that the technique cannot lead to a lower bound beyond $\Omega(\sqrt n)$.

In our proof, 
rather than requiring that many processes cover the same register, we limit the number of processes covering the same register to three.
In particular, we define a $(3,k)$-configuration to be one where $k$ processes are covering registers, but no register is covered by more than three of them.
Using an argument reminiscent of that used by Burns and Lynch~\cite{BurnsL93}, 
we show that if there is an execution that leads to some $(3,k)$-configuration, we can find a (much longer) execution, 
during which at least two $(3,k)$-configurations $C_1$ and $C_2$ are encountered that are similar in the sense that in both configurations each register is covered by the same number of processes.
In addition, the execution $(C_1;\sigma)$ that leads from $C_1$ to $C_2$ starts with three block-writes to the registers that are covered by three processes, each.
We then apply Lemma~\ref{lemma:EFR2008a-main} to see that we can insert a $p$-only schedule for some unused process $p$ 
into the schedule $\sigma$ after one of the block-writes to get the new schedule $\sigma'$, 
such that at the end of the execution $(C_1;\sigma')$ process $p$ is poised to write outside of the registers that are 3-covered in $C_1$.
Since the other two block-writes overwrite $p$'s trace in $(C_1;\sigma')$,
no process other than $p$ can distinguish between $\sigma'(C_1)$ and $\sigma(C_1)=C_2$.
It follows that in $\sigma'(C_1)$ process $p$ covers a register that is covered by at most 2 other processes.
Hence, we have obtained a $(3,k+1)$-configuration.
We can do this for $k\leq n/2$, so in the end we obtain a $(3,\lfloor{n/2}\rfloor)$-configuration.
Clearly, this means that the number of registers is at least $\lfloor{n/6}\rfloor$.

The signature of a configuration $C$, denoted \sig{C}, 
is a tuple $(c_1,c_2,\ldots , c_m)$ where every $c_i$ is the number of processes  covering the $i$-th register in $C$.
The set of registers whose corresponding entry in $\sig{C}$ is equal to $3$ is  denoted $\tre{3}{C}$.
(In terms of signatures, a configuration $C$ is a $(3,k)$-configuration if $\sig{C}=(c_1,c_2,\ldots , c_m)$ satisfies $\sum_{i=1}^{m} c_i  =  k$ and $c_i \leq 3$ for every $1 \leq i \leq m$.) 
Notice that in any $(3,k)$-configuration there are at least $\lceil k/3 \rceil $ registers covered.
Configuration $C$ is \emph{quiescent} if in $C$ no process has started but not finished executing a \getTS{} or \compare{} call.

\begin{lemma}\label{lemma:signature}\samepage 
Let $P$ be an arbitrary set of processes. 
Suppose for every reachable quiescent configuration $D$ there exists a $P$-only schedule $\sigma$ such that $\sigma(D)$ is a $(3,k)$-configuration.
Then for any  quiescent configuration $D$, there are two $(3,k)$-configurations $C_0$ and $C_1$,  and $P$-only schedules $\gamma_0$, $\gamma_1$, and $\eta$ such that:

\begin{description}
\item[(a)] $\gamma_0(D)= C_0$,
\item[(b)] $\gamma_1(C_0)= C_1$,
\item[(c)] $\sig{C_0}= \sig{C_1}$, and
\item[(d)] $\gamma_1=\pi_{B_0}\pi_{B_1}\pi_{B_2}\eta$, where $B_0,B_1$ and $B_2$ are disjoint sets of processes each covering $\tre{3}{C_0}$.
 \end{description}
\end{lemma}

\begin{proof}
We inductively  define an infinite sequence of schedules $\lambda_0,\delta_0,\lambda_1,\delta_1,\ldots,\lambda_{i},\delta_{i},\ldots$
and reachable $(3,k)$-configurations $E_0, E_1,E_2,\dots$, where
$E_{i+1}=\lambda_i\delta_i(E_{i})$, as follows. 
$E_0$ is the $(3,k)$-configuration $\sigma(D)$ guaranteed by the hypothesis of the lemma.
Let $B_{0,i},B_{1,i}$ and $B_{2,i}$ be disjoint sets of processes each covering $\tre{3}{E_{i}}$. 
Execution $(E_{i};\pi_{B_{0,i}}\pi_{B_{1,i}}\pi_{B_{2,i}})$ consists of three consecutive block-writes to $\tre{3}{E_{i}}$ 
by the processes in $B_{0,i}$, $B_{1,i}$, and $B_{2,i}$, respectively. 
Schedule $\lambda_{i}$ is the concatenation of the sequence of permutations $\pi_{B_{0,i}}\pi_{B_{1,i}}\pi_{B_{2,i}}$ 
and some $P$-only schedule $\alpha_{i}$
in which every process in $P$ with a pending operation, 
finishes that pending operation.  
Thus, configuration $\lambda_{i}(E_{i}) = \pi_{B_{0,i}}\pi_{B_{1,i}}\pi_{B_{2,i}}\alpha_{i}(E_{i})$ is  quiescent.
So by the hypothesis 
there exists a schedule $\delta_{i}$ such that $E_{i+1} = \lambda_{i}\delta_{i}(E_{i})$ is again a $(3,k)$-configuration. 

Since the set of signatures is finite, there are two indices $j<k$, such that $\sig{E_j}=\sig{E_k}$. 
Fix two such indices $j$ and $k$. 
Let  $\gamma_0=\sigma\lambda_0\delta_0\lambda_1\delta_1\lambda_2\delta_2\ldots\lambda_{j-1}\delta_{j-1}$ and 
$\gamma_1=\lambda_{j}\delta$ where
$\delta=\delta_{j}\lambda_{j+1}\delta_{j+1}\ldots\lambda_{k-1}\delta_{k-1}$.
Furthermore, let $C_0 = \gamma_0(D)$ and $C_1 = \gamma_1(C_0)$.
By definition, the configurations $C_0$ and $C_1$ satisfy (a) and (b).
Moreover, by construction $C_0=E_j$ and $C_1=E_k$ and since $\sig{E_j}=\sig{E_k}$, (c) is satisfied.
Finally, let $\eta= \alpha_j\delta$. Then, $\gamma_1=\pi_{B_{0,j}}\pi_{B_{1,j}}\pi_{B_{2,j}}\eta$, 
where $B_{0,j},B_{1,j},B_{2,j}$ are disjoint sets of processes each covering $\tre{3}{E_j}=\tre{3}{C_0}$.
This proves (d). 
\end{proof}

Let $\mathcal{P}_k$ denote the set $\{p_1,\dots,p_k\}$ and $P_0$ denote the emptyset of processes.

\begin{lemma}\label{lemma:long-lived-induction}
  For every $0\leq k \leq \lfloor n/2 \rfloor$ and for every reachable quiescent configuration $D$, 
  there exists a $\mathcal{P}_{2k}$-only schedule $\sigma_k$ such that $\sigma_k(D)$ is a $(3,k)$-configuration. 
\end{lemma}

\begin{proof}
The proof is by induction on $k$.
For $k=0$ the claim is immediate by choosing $\sigma_0$ to be the empty schedule. 

Let $k \geq 1$, and let $D$ be an arbitrary reachable quiescent configuration. 
By the induction hypothesis, for every reachable quiescent configuration $C$, there
exists a $\mathcal{P}_{2k-2}$-only schedule $\sigma_{k-1}$, such that $\sigma_{k-1}(C)$ 
is a $(3,k-1)$-configuration.  
Hence, by Lemma~\ref{lemma:signature} with $P=\mathcal{P}_{2k-2}$ 
there are two reachable configurations $C_0$ and $C_1$, and $\mathcal{P}_{2k-2}$-only schedules $\gamma_0, \gamma_1$, and $\eta$,
such that $\gamma_0(D)=C_0$, $\gamma_1(C_0)=C_1$, $\sig{C_0}= \sig{C_1}$, and $\gamma_1=\pi_{B_0}\pi_{B_1}\pi_{B_2}\eta$, where  $B_0,B_1$ and $B_2$ are disjoint sets of processes, each covering $\tre{3}{C_0}$.

Consider the two processes $p_{2k-1}$ and $p_{2k}$. 
For $i\in\{0,1\}$, let $\alpha_i$ be a $\{p_{2k-i}\}$-only schedule 
such that in execution $(\pi_{B_i}(C_0); \alpha_i)$,  $p_{2k-i}$ performs a complete \getTS{} instance.
According to Lemma~\ref{lemma:EFR2008a-main}, there exists $i\in \{0,1\}$, 
such that $p_{2k-i}$ writes to some register not in $\tre{3}{C_0}$ 
during the execution $(\pi_{B_i}(C_0);\alpha_i)$.
(Note that whether $i=0$ or $i=1$ depends on $C_0$.) 
Let $r$ be the first register not in $\tre{3}{C_0}$  to which $p_{2k-i}$ writes to in $(\pi_{B_i}(C_0);\alpha_i)$. 
Since $\sig{C_0}= \sig{C_{1}}$, we have $r\notin \tre{3}{C_{1}}$, 
and thus $r$ is covered by at most two processes in $C_0$ as well as in $C_1$.

Let $\lambda$ be the shortest prefix of $\alpha_i$  such that $p_{2k-i}$ is about to write to $r$ in $\pi_{B_i}\lambda(C_0)$. 
Since $p_{2k-i}$ does not participate in schedule $\pi_{B_{1-i}}\pi_{B_2}\eta$, it is also covering $r$ in the configuration $\pi_{B_i}\lambda\pi_{B_{1-i}}\pi_{B_2}\eta(C_0)$.
Configurations $\pi_{B_i}\pi_{B_{1-i}}\pi_{B_2}(C_0)$ and $\pi_{B_{1-i}}\pi_{B_i}\pi_{B_2}(C_0)$ are indistinguishable to all processes; 
therefore, $\pi_{B_i}\pi_{B_{1-i}}\pi_{B_2}\eta(C_0)=C_1$. 
Moreover, since $C_1=\pi_{B_0}\pi_{B_1}\pi_{B_2}\eta(C_0)$ is indistinguishable from $\pi_{B_i}\lambda\pi_{B_{1-i}}\pi_{B_2}\eta(C_0)$ to every process except  $p_{2k-i}$, all processes other than $p_{2k-i}$ cover the same registers in $C_1$ as in $\pi_{B_i}\lambda\pi_{B_{1-i}}\pi_{B_2}\eta(C_0)$. 
Since $p_{2k-i}$ covers $r$ in this configuration, and $r$ is covered by at most 2 other processes, $\pi_{B_i}\lambda\pi_{B_{1-i}}\pi_{B_j}\eta(C_0)$ is a $(3,k)$-configuration.
\end{proof}

Lemma~\ref{lemma:long-lived-induction} shows that in  any  long-lived unbounded timestamp implementation that satisfies 
non-deterministic solo-termination
there exists a reachable $(3,\lfloor n/2 \rfloor)$-configuration.
Clearly, at least $\lfloor n/6 \rfloor > n/6-1$ registers are covered in this configuration.
This proves Theorem~\ref{theorem:long-lived}.

\section{A Space Lower Bound for One-shot Timestamps}
\label{section:one-shot-lwb}

It seems natural to imagine that $n$ registers would be required to construct a timestamp system 
for $n$ processes.  
But this is not the case for some restricted versions of the problem.  
For example, if the timestamps are not required to come from a nowhere dense set, 
then, as shown by Ellen, Fatourou and Ruppert \cite{EFR2008a}, 
$n-1$ registers suffice. 
We show that another instance is when each process is restricted to at most one call to the \getTS{} method. 
In this case $\Theta(\sqrt{n})$ registers are necessary and sufficient. 
This section contains the space lower bound.  
Section~\ref{section:one-shot-ub} contains the algorithm that shows that this lower bound 
is asymptotically tight.

Our lower bound proof relies on Lemma \ref{lemma:revisedCore}, 
the proof of which uses Lemma \ref{lemma:EFR2008a-main} inductively.
Given four disjoint sets of processes $B_1, B_2, B_3, U$ 
such that processes in $B_1, B_2, B_3$ cover a set of registers $R$, 
then, according to Lemma \ref{lemma:EFR2008a-main},
for any partition
of $U$ into $V_1$ and $V_2$, 
either all the processes in $V_1$ or all the processes in $V_2$ 
can be made to cover some register outside of $R$.
By choosing $V_1$ and $V_2$ to have sizes differing by at most one, 
Lemma \ref{lemma:EFR2008a-main} can be used to ensure that essentially half 
of the processes in $V_1 \cup V_2$ must write outside of $R$.
 
We strengthen this idea by using Lemma \ref{lemma:EFR2008a-main} inductively 
to construct an execution such that 
all but one of the processes in $U$ that have not initiated any operation 
can be made to cover some register outside of the set of registers $R$.
Let \procs{\sigma} denote the set of the processes taking steps in schedule $\sigma$. 
A process is \emph{idle} in configuration $C$ 
if it is in its initial state in $C$;
the set of all such processes is denoted \idle{C}.

\begin{lemma}
\label{lemma:revisedCore}

Let $C$ be a reachable configuration of a one-shot timestamp implementation from registers 
that satisfies non-deterministic solo-termination. 
Let $B_0$, $B_1$, $B_2$, $U$ be disjoint sets of processes where in $C$ 
each of $B_0$, $B_1$ and $B_2$ cover a set $R$ of registers 
and $U \subseteq \idle{C}$, with $|U| \geq 2$.
Then there is a schedule $ \beta\sigma\beta'\sigma'$ satisfying:

\begin{itemize}
  \item[(a)] 
       $\{\beta, \beta' \} = \{ \pi_{B_0}, \pi_{B_1} \}$;
  \item[(b)] 
       In configuration $ \beta\sigma\beta'\sigma'(C)$ all processes in 
       \procs{\sigma} and \procs{\sigma'} cover a register outside of $R$;
  \item[(c)] 
        $\procs{\sigma} \cup  \procs{\sigma'}  \subsetneq  U $.
  \item[(d)] 
        $\left| \procs{\sigma} \right| + \left| \procs{\sigma'} \right| = \left| U \right| - 1$;
  \item[(e)] 
        $ \left| \procs{\sigma} \right| \geq  \lfloor \left| U \right| / 2 \rfloor  \geq 
            \left| \procs{\sigma'} \right| $;
  \item[(f)] 
        $\sigma$ and $\sigma'$ are concatenations of solo schedules by distinct processes in $U$. 
\end{itemize}
\end{lemma}

\begin{proof}
Let $U = \{ p_0, \dots, p_{m} \}$, where $m \geq 1$ (because $|U| \geq 2)$. 
For each $1 \leq k \leq m$, 
we first inductively construct schedules $\delta_0^k$ and $\delta_1^k$ such that
\begin{itemize}
\item 
      \procs{\delta_0^k} and \procs{\delta_1^k} form a partition of $ \{p_0, \ldots p_k \}$;
\item 
      for $i\in \{0,1 \}$, in execution $(C; \pi_{B_i} \delta_i^k)$:
\begin{itemize}
\item
      each process in \procs{\delta_i^k} initiates exactly one instance of \getTS{};
\item
       exactly one \getTS{} method completes, and this \getTS{} is by the last process in $\delta_i^k$; 
\item  
      no process except possibly the last that occurs in $\delta_i^k$ writes outside of $R$; 
\end{itemize}
\item  
      for $i\in \{0,1 \}$, in configuration $\pi_{B_i} \delta_i^k(C)$
      every process in \procs{\delta_i^k} except possibly the last that occurs in $\delta_i^k$ 
      cover a register outside of $R$. 
\end{itemize}

For $i \in \{0,1\}$, 
let  $\delta_i^1$ be a $p_i$-only schedule in which process $p_i$ 
performs a complete \getTS{} instance in the execution $(C; \pi_{B_i} \delta_i^1)$.
Such a schedule $\delta_i^1$ exists because $p_0$ and $p_1$ are in \idle{C}.
This immediately satisfies the base case, $k=1$.

For $i \in \{0,1 \}$, 
suppose that $\delta_i^k$ are constructed as required, 
and let $q_i$ denote the last process in $\delta_i^k$.
Since execution $(C; \pi_{B_i} \delta_i^k)$  contains a complete \getTS{} by $q_i$,
and no process in $\delta_i^k$ before $q_i$ writes outside of $R$ in $(C; \pi_{B_i} \delta_i^k)$,
Lemma \ref{lemma:EFR2008a-main} implies that either 
$q_0$  in execution $(\pi_{B_0}(C);\delta_0^k)$ or
$q_1$  in execution $(\pi_{B_1}(C);\delta_1^k)$ 
must write outside of $R$. 
Choose such a $j \in \{0,1\}$ such that process $q_j$ does write outside of $R$ in $(\pi_{B_j}(C);\delta_j^k)$. 
First truncate the schedule $\delta_j^k$ , to, say, $\alpha_j^k $, 
by deleting a  suffix of the solo schedule of $q_j$ so that, instead of completing its \getTS{} method,
$q_j$ is paused at the earliest point such that at the end of the execution $(\pi_{B_j}(C);\alpha_j^k)$,
$q_j$ covers a register outside of $R$.
Now append to $\alpha_j^k$, a $p_{k+1}$-only schedule $\sigma_{k+1}$ so that 
the execution $(\pi_{B_j}(C);\alpha_j^k \sigma_{k+1})$ contains a complete \getTS{} method by $p_{k+1}$.
Define $\delta_j^{k+1}$ to be $\alpha_j^k \sigma_{k+1}$ and
$\delta_{1-j}^{k+1}$ to be $\delta_{1-j}^k$.
The claimed construction now holds for $k+1$.

Therefore, we can construct two schedules, $\delta_{i}^{m}$ for $i \in \{0,1\}$, 
that together contain all the processes of $U$ and where each is a concatenation of 
distinct solo-executions.  
Furthermore, each of the executions $(C; \pi_{B_i} \delta_i^m)$ contains exactly
one complete \getTS{} by the last process in the schedule $\delta_i^m$, 
and no other process writes outside of $R$. 
Therefore, applying Lemma \ref{lemma:EFR2008a-main} one more time, 
for a $j\in \{0,1 \}$, 
in the execution $(C; \pi_{B_j} \delta_j^m)$, the last process in $\delta_j^m$ 
must write outside of $R$.
Let $\sigma_j $ be the schedule $\delta_j^m$ truncated to the first point such that at the end of 
execution $(C; \pi_{B_j} \sigma_j)$ this last process covers a register outside of $R$.
Let $\sigma_{1-j}$ be the schedule $\delta_{1-j}^m$ truncated to remove the entire schedule 
of its last process.

Now relabel the members of $\{\pi_{B_0} \sigma_0, \pi_{B_1} \sigma_1 \}$ 
to have distinct names in $\{\beta \sigma, \beta' \sigma' \}$ in such a way that
the two schedules $\sigma_0$ and $\sigma_1 $ are renamed with distinct names in  $\{\sigma, \sigma' \}$ 
and satisfy $\left|\procs{\sigma} \right| \geq  \left|\procs{\sigma'} \right|$.    
By construction, 
$\procs{\sigma_0} $ and $ \procs{\sigma_1}$ do not intersect;
each is a subset of $U$; and together they contain all but 1 of the members of $U$.
Also, by construction, each of  $\sigma_0 $ and $\sigma_1 $ are concatenations of solo executions.
When combined with the relabeling, this establishes (a), (c), (d), (e) and (f).

Since \procs{\sigma_0} and \procs{\sigma_1} are disjoint sets,
and since no process writes outside of $R$ 
in the execution $(C; \pi_{B_i} \sigma_i)$  for $i \in \{0,1\}$,
and since each block write obliterates all writes to $R$,
configurations $\pi_{B_i}(C)$ and $\pi_{B_{1-i}} \sigma_{1-i} \pi_{B_i} (C) $ are
indistinguishable to $\procs{\sigma_i}$.
So each process in \procs{\sigma_i} covers the same register in  
$\pi_{B_i}\sigma_i(C)$ as it does in $\pi_{B_{1-i}}\sigma_{1-i}\pi_{B_i}\sigma_i(C)$
and as it does in $\pi_{B_i}\sigma_i\pi_{B_{1-i}}\sigma_{1-i}(C)$.
Consequently, in both 
$ \pi_{B_0}\sigma_0\pi_{B_1}\sigma_1(C)$ and 
$ \pi_{B_1}\sigma_1\pi_{B_0}\sigma_0(C)$ 
each of the $m-1$ processes in $\procs{\sigma_0} \cup \procs{\sigma_1}$ covers a register not in $R$.
This, combined with the relabeling, establishes (b).
\end{proof}

Lemma \ref{lemma:revisedCore} is the principle tool for our space lower bound for one-shot timestamps.
To describe the structure of the proof we use the following definitions.
Let $m= \lfloor \sqrt{2n}\rfloor $.
Assume that the set of all registers, denoted \allRegs, has size at most $m$ since otherwise we are done.
Define the \emph{ordered-signature of a configuration $C$}, denoted \ordsig{C}, 
to be the $m$-tuple $(s_{1},s_{2}, \ldots, s_{m})$ where $s_{i} \geq s_{i+1}$, 
and there is a permutation $\alpha$ of $\{1, \ldots , m \}$
such that for $1 \leq i \leq m$,
$s_i$ processes are covering the $\alpha_i$-th register.
(The ordered-signature of a configuration is just its signature 
with the entries of the $m$-tuple reordered so that they are non-increasing.
If only $k<m $ registers exist 
then $s_{k+1} = s_{k+2} = \dots = s_m =0$.) 
A configuration $C$ with $\ordsig{C} = (s_{1}, \ldots, s_{m})$,
is \emph{$\ell$-constrained} if $s_c \leq \ell -c $ for every $ 1 \leq c \leq \ell$.
A configuration $C$ is \emph{\full{j}{k}} if there is a set $R$ of registers
such that $|R| = j$ and in $C$ each register in $R$ is covered by at least $k$ processes.
If $C$ is \full{j}{k}, \fullRegs{j}{k}{C} denotes a set of such registers,
otherwise \fullRegs{j}{k}{C} is undefined.

If $C$ is \full{j}{k} where $k\geq 3$,  
and there are  $u \geq 2$ processes that are idle in $C$,
then Lemma \ref{lemma:revisedCore} can be applied with 
$B_0$, $B_1$, $B_2$ any 3 disjoint sets each covering \fullRegs{j}{k}{C},
so that for any $1 \leq v \leq u-1$, $v$ processes can be made to cover registers outside of 
\fullRegs{j}{k}{C} using at most 2 block writes to \fullRegs{j}{k}{C}.

We use this idea repeatedly to construct an execution that visits a sequence of 
configurations, say $C_1, ..., C_{\text{last}}$
such that the set of registers covered in $C_{i+1}$
is a superset of the set covered in $C_i$
until eventually a configuration $C_{\text{last}}$ is reached in which at least 
$m -\log{n} $ registers are covered.

Intuition for our construction is aided by a geometric representation of configurations.  
Configuration $C$ with $\ordsig{C} = (s_{1},s_{2}, \ldots, s_{m})$ 
is represented on a grid of cells where, in each column $c$, $1\leq c \leq m$, 
the lowest $s_c$ cells are shaded.
Thus each register corresponds to a column in the grid, 
but this correspondence can change for different configurations.
With this interpretation, 
each shaded cell in column $c$ represents a process covering the register corresponding to $c$. 
If the configuration is $l$-constrained, 
the shading in each column remains below the stepped diagonal that starts at height $l-1$ in the grid.
The configuration is \full{j}{k} if in column $j$ (and hence in all columns 1 through $j$)
the height of the shaded cells is at least $k$.

An overview of the construction is as follows.
We first achieve an $m$-constrained \full{j}{m-j} configuration for some $j \geq 1$
as shown in Figure~\ref{figure:gridOne}.
\begin{figure}[h]
	\centering
		\includegraphics[width=3.4in]{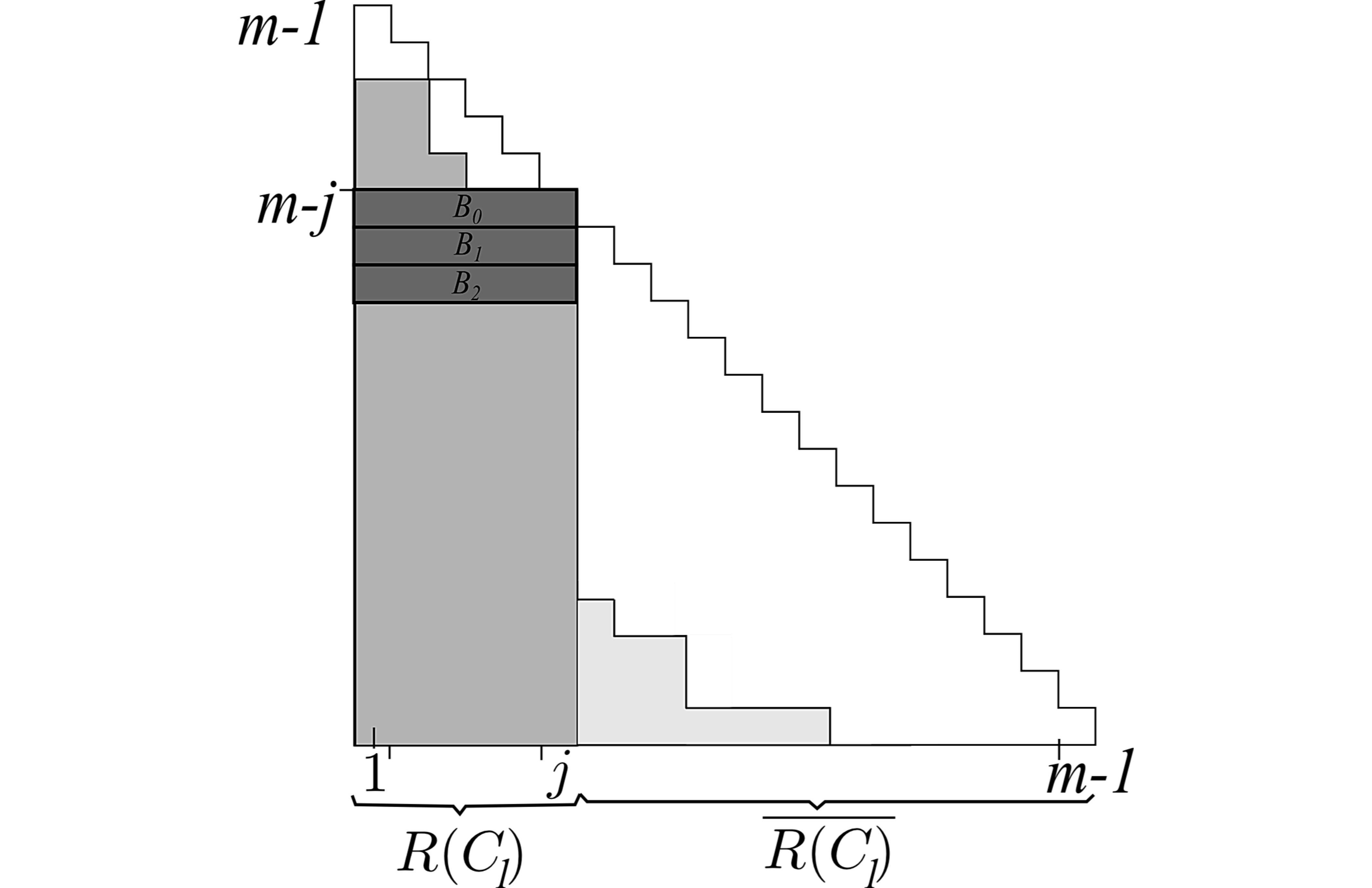}
	\label{figure:step1}
	\caption{Configuration $C_1$ must have a column $j$ that reaches to the diagonal. 
        Hence there are $j$ registers each covered with $m-j$ processes.}
\label{figure:gridOne}
\end{figure}

Given some $\ell$-constrained \full{j}{\ell-j} configuration, 
(such as shown in Figure~\ref{figure:gridOne} with $m=\ell$)
and provided $\ell-j$ is at least 3,  
we can apply Lemma \ref{lemma:revisedCore} using 3 disjoint sets of processes each occupying cells in columns 1 through $j$
for the sets $B_0$, $B_1$ and $B_2$. 
Then, one at a time, idle processes can be made to occupy cells in columns $j+1$ through $m$.
We will maintain the invariant that the number of idle processes is always greater than the number of unshaded cells
that are under the stepped diagonal and in columns $j+1$ through $m$. 
Because of this invariant, we can be sure to reach a configuration $C'$ where, for the first time,
(when the columns $j+1$ through $m$ are rearranged in order of non-increasing number of occupants) 
some column $j' \geq j+1 $ gets $\ell-j'$ occupants. 
During this execution 
the block writes reduced the height of the shaded cells in columns $1$ through $j$ by one or two. 
If only one block write happened during this execution, or if $j' \geq j+2$,
$C'$ is again an $\ell$-constrained \full{j'}{\ell-j'} configuration
(Case 1 of Figure~\ref{figure:gridTwo}).

\begin{figure}[h]
\begin{center}
$\begin{array}{c@{\hspace{0.2in}}c}
        \epsfxsize=2.0in
\epsffile{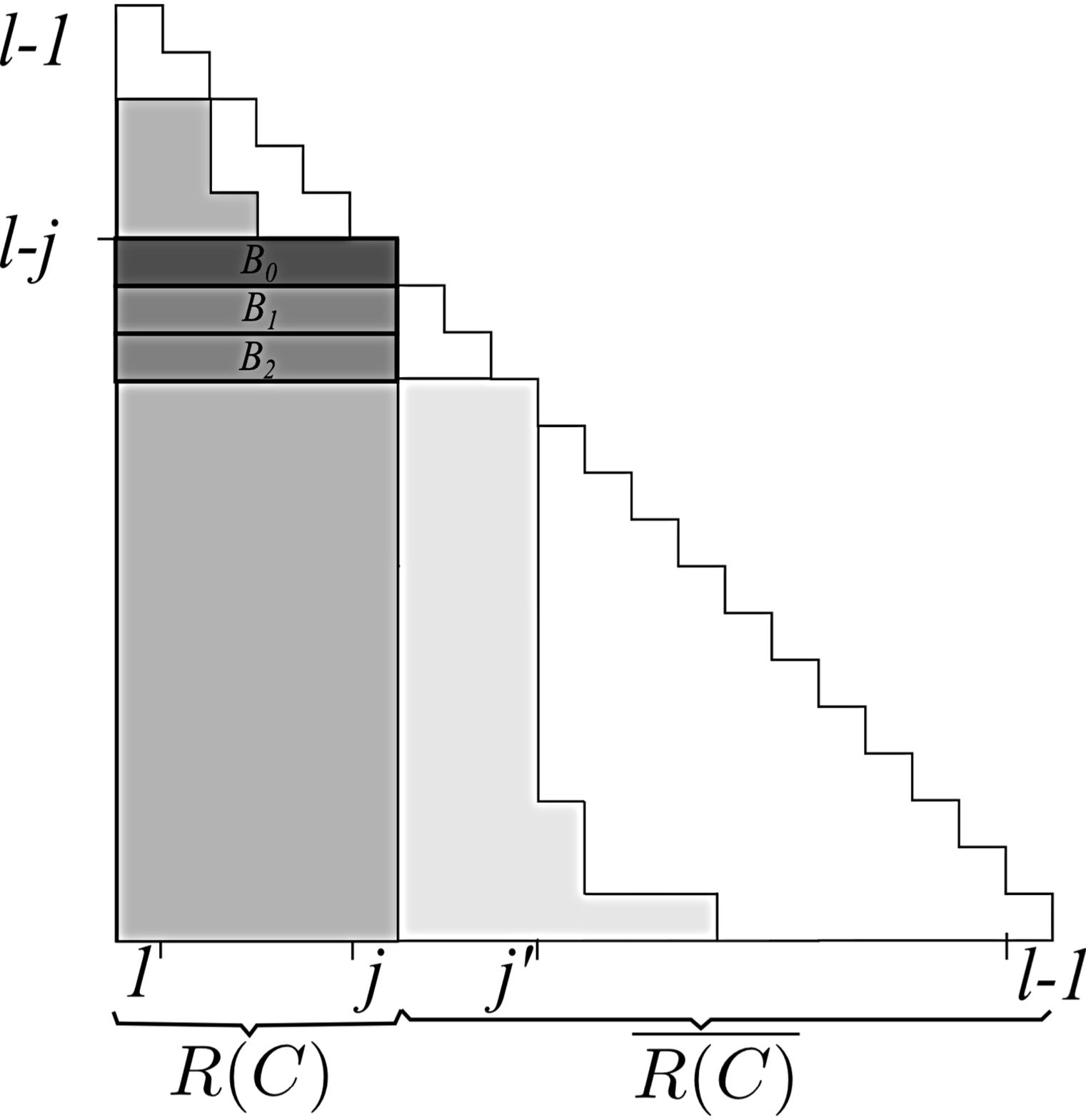} &
        \epsfxsize=2.0in
        \epsffile{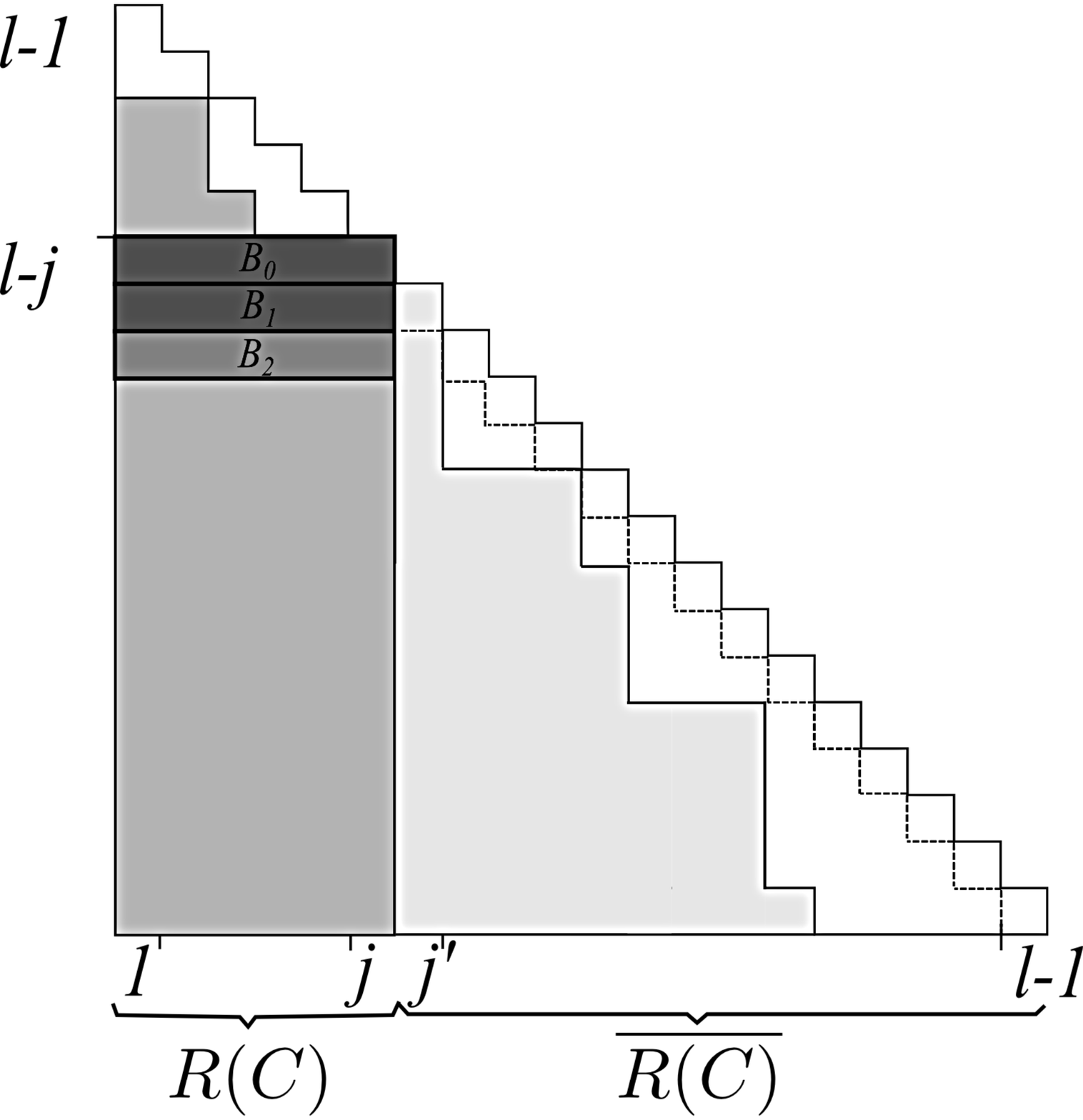} \\ [0.0cm] \mbox{Case 1} & \mbox{Case 2} 
\end{array}$ 
\end{center} 
\caption{After the block-write, processes are run until some new column $j'$ reaches the diagonal and thus has height $\ell-j'$. 
Case 1: columns 1 through j still have height at least $\ell-j'$. 
Case 2: the diagonal is reached at column $j+1$ after two block writes.   
This can only happen if at least half of the unshaded space in columns $j+1$ through $m$ became shaded.}
 \label{figure:gridTwo} 
\end{figure}

The only other case is when both block writes were used to achieve $C'$ and $j' = j+1$
(Case 2 of Figure~\ref{figure:gridTwo}).
Then $C'$ is an $(\ell-1)$-constrained \full{j'}{\ell-1-j'}  configuration.
In this case, however, at least half of the idle processes have moved to occupy cells in columns $j+1$ through $m$.
So this reduction by one in the stepped boundary of the grid can only happen $\log n$ times. 
Thus, each repetition of this construction creates a $(m-s)$-constrained \full{k}{m-k-s} configuration
where $s \in O(\log n)$.
The construction can be repeated until either there are fewer than 2 idle processes or $m-k-s <3$.
In both cases at least $m-s = \sqrt{2n} - O(\log n)$ registers are covered.

The rest of this section contains the details of this construction, 
which provides the proof of Theorem \ref{theorem:one-timeLbd}.
We assume that $n \geq 3$ since otherwise the theorem is trivially correct.
For configuration $C$, 
and a set of registers $R \subseteq \allRegs$, 
\poised{C}{R} denotes the processes that are covering some register in $R$. 
For any set of registers $R$, $\overline{R}$ denotes the set $\allRegs \setminus R$.

The construction is inductive, starting with the initial configuration $C_0$.
Initialize $j_0=0$, $\ell_0 =m $ and  $R_0 = \emptyset$.

In $C_0$, no register is covered. 
Set $B_0 =B_1=B_2 =\emptyset$, let $U$ be the set of all $n$ processes
and apply Lemma \ref{lemma:revisedCore}. 
Because $\pi_{B_i}$ is the empty schedule for $i \in \{0,1,2\}$,
the schedule produced is $ \sigma\sigma'$ and 
$n-1$ processes cover some register in configuration $ \sigma\sigma'(C_0)$.
Let $\ordsig{\sigma\sigma'(C_0)} = (s_{1},s_{2}, \ldots, s_{m})$.
If $s_{m} \geq 1$, we are done since then $m$ registers are covered.
Therefore, assume $s_{m} = 0$,  and suppose $s_c \leq m-c-1$ for each $c \leq m -1 $.
Then, $\sum_{c=1}^{m}  s_c  \leq \sum_{c=1}^{m-1}  (m-c-1)  + s_m = (m-1)(m-2)/2 + 0  < n-1$,
which is impossible.
Hence, there is at least one $j \leq m-1$ satisfying $s_j \geq m-j$
and $ \sigma\sigma'(C_0)$ is therefore \full{j}{m-j}.
Let $\gamma_1$ be the shortest prefix of $\sigma\sigma'$ so that
there is such a value, which we label $j_1$, satisfying $\gamma_1(C_0)$ is a \full{j_1}{m-j_1} configuration. 
Configuration $\gamma_1(C_0)$ must also be $m$-constrained, because otherwise,
there is some index $i$ such that $i$ registers are covered by at least $m-i+1$ processes in configuration $\gamma_1(C_0)$. 
But then there is a proper prefix, $\alpha$, of $\gamma_1$ such that $\alpha(C_0)$ is a \full{i}{m-i} configuration, for some $i$.
Define $C_1 = \gamma_1(C_0)$,  $\ell_1 = m$ and  $R_1 = \fullRegs{j_1}{\ell_1-j_1}{C_1}$.

In execution $(C_0; \gamma_1)$, each process $p$ in \procs{\gamma_1} leaves the set \idle{C_0} and
performs a $p$-only execution until it is paused when it covers a register.
Therefore, $| \poised{C_1}{\allRegs} | +| \idle{C_1} | = n$. 
At most  $\sum_{c=1}^{j_1} (m-c -1)  + 1 $ processes cover registers in $R_1$.
The remainder of at least $n- (\sum_{c=1}^{j_1} (m-c -1)  +1) >  \sum_{c=j_1+1 }^{m} (m-c)   $ 
processes are either still idle or are covering registers in $\overline{R_1}$. 
So for $i=1$, the following \emph{construction invariant} holds: 
\begin{enumerate}
\item[(a)]
$C_i = \gamma_{i}(C_{i-1})$ 
\item[(b)]
$R_{i-1} \subsetneq R_i$ 
\item[(c)]
$|\poised{C_i}{\overline{R_i}}| + |\idle{C_i}| -1 \geq  \sum_{c=j_i+1 }^{m} (m-c) $
\item[(d)]
$j_i \geq  j_{i-1} + 1$ and  $\ell_i \in \{\ell_{i-1}, \ell_{i-1}-1 \}$ and $\ell_i \leq m$  
\item[(e)]
$C_i$ is a \full{j_i}{\ell_i-j_i} configuration with $R_i = \fullRegs{j_i}{\ell_i-j_i}{C_i}$.
\end{enumerate}

Now suppose that a sequence of tuples $(\gamma_1, C_1, j_1, \ell_1, R_1), \ldots , (\gamma_k, C_{k}, j_{k}, \ell_{k}, R_k) $ 
has been built so that the construction invariant holds for each.
If $\ell_k-j_k \geq 3 $ and $|\idle{C_{k}}| \geq 2$ 
then let $B_0, B_1, B_2$ to be disjoint sets of processes, such that each covers $R_k$  and each has size $|R_k|$, 
and let $U = \idle{C_k}$.
According to Lemma \ref{lemma:revisedCore}, there is a  schedule $ \beta\sigma\beta'\sigma'$ 
satisfying: 
\begin{itemize}
\item
$\beta $ and $\beta'$ are block writes by $B_0$ and $B_1$,
\item
$\sigma$ and $\sigma'$ are concatenations of solo schedules by distinct processes in \idle{C_k},
\item
$|\procs{\sigma}| \geq \lfloor |\idle{C_k}| / 2 \rfloor  $
\item
$|\procs{\sigma}| + |\procs{\sigma'}| =  |\idle{C_k}| -1  $, and
\item
in configuration $ \beta\sigma\beta'\sigma'(C_k)$ all processes in 
       \procs{\sigma} and \procs{\sigma'} cover a register in $\overline{R_k}$.
\end{itemize}

So, combining with (c) of the construction invariant, 
$|\poised{\beta\sigma\beta'\sigma'(C_k)}{\overline{R_k}}|  \geq  \sum_{c=j_k+1 }^{m} (m-c)$. 
Hence, there must be a non-empty set of registers $Q' \subseteq \overline{R_k}$ such 
that each is covered by at least $\ell_k -j_k - |Q'| $ processes. 
Let $\gamma_{k+1}$ be the shortest prefix of $\beta\sigma\beta'\sigma'$ such that
there is such a $Q'$, which we call $Q$, in $\gamma_{k+1}(C_k)$ and let $\nu_k = |Q| $, where $\nu_k\geq 1$. 
Define $C_{k+1} = \gamma_{k+1}(C_k)$, $R_{k+1} = Q \union R_k$, and $j_{k+1} = \nu_k + j_k$.
Then $|R_{k+1}| = \nu_k + j_k= j_{k+1}$.
During execution $(C_k;\gamma_{k+1})$, processes that leave $\idle{C_k}$ pause when they cover a register 
in $\overline{R_k}$.
So 
$|\poised{C_k}{\overline{R_k}}| + |\idle{C_k}| = |\poised{\gamma_{k+1}(C_k)}{\overline{R_k}}| + |\idle{\gamma_{k+1}(C_k)}| $.
Therefore, by (c), 
$|\poised{C_{k+1}}{\overline{R_k}}| + |\idle{C_{k+1}}| -1 \geq  \sum_{c=j_k+1 }^{m} (m-c) $.
Furthermore, since $\gamma_{k+1}$ is chosen to be as short as possible,
$$|\poised{C_{k+1}}{Q}| \leq \sum_{c=j_k+1 }^{j_k+\nu_k} (\ell_k-c-1) +1  < \sum_{c=j_k+1 }^{j_k+\nu_k} (m-c) $$
Therefore, 
$$|\poised{C_{k+1}}{\overline{R_k} \setminus Q}| + |\idle{C_{k+1}}| -1 >  \sum_{c=j_k+1 }^{m} (m-c) - \sum_{c=j_k+1 }^{j_k+\nu_k} (m-c)
= \sum_{c=j_{k+1} +1  }^{m} (m-c) $$
Thus (a), (b), and (c) of the construction invariant hold for $k+1$.
For parts (d) and (e) there are two cases.

\emph{Case 1: $\gamma_{k+1}$ is a prefix of $\beta\sigma$ or $\nu_k \geq 2$}.
If $\gamma_{k+1}$ is a prefix of $\beta\sigma$ then there is only one block write to $R_k$.
So in $C_{k+1}$,  
each of the $j_k$ registers in $R_k$ remains covered by at least $\ell_k-j_k -1$ processes 
and each of the $\nu_k$ registers in $Q$ is covered by at least $\ell_k-j_k - \nu_k \leq \ell_k-j_k -1$ processes.
If  $\nu_k \geq 2$, then in $C_{k+1}$,  
each of the $j_k$ registers in $R_k$ remains covered by at least $\ell_k-j_k -2$ processes 
and each of the $\nu_k$ registers  in $Q$ is covered by at least $\ell_k-j_k - \nu_k \leq \ell_k-j_k -2$ processes.
So in either situation, each of the $j_{k+1} = j_k +\nu_k$ registers in $R_{k+1}$ is covered by at least 
$\ell_k-j_k - \nu_k= \ell_k - j_{k+1}$ processes.
Therefore, setting $\ell_{k+1} = \ell_k$ we have that
$C_{k+1}$ is a \full{j_{k+1}}{\ell_{k+1}-j_{k+1}} configuration and the construction invariant holds.

\emph{Case 2: $\nu_k=1$ and $\gamma_{k+1}$ is not a prefix of $\beta\sigma$}.
In this case there are two block writes to $R_k$. 
So in $\gamma_{k+1}(C_k)$, 
each register in $R_k$ remains covered by only $\ell_k-j_k-2$ processes,
which is one fewer than the number of processes covering the single register in $Q$. 
Since $j_{k+1} = j_k +1$, we can set $\ell_{k+1} = \ell_k -1$ to ensure that 
$C_{k+1}$ is a \full{j_{k+1}}{\ell_{k+1}-j_{k+1}} configuration.
So, again the construction invariant holds.

This  construction ends in a configuration $C_{\mathrm{last}}$ where 
either $\ell_{\mathrm{last}}-j_{\mathrm{last}}  \leq 2 $ or $|\idle{C_{\mathrm{last}}}|=1$,
since in either case Lemma~\ref{lemma:revisedCore} can no longer be applied.
Clearly, $C_{\mathrm{last}} = \gamma_1, \gamma_2, \ldots , \gamma_{\mathrm{last}} (C_0)$ 
so $C_{\mathrm{last}}$ is reachable. 
Since, in $C_{\mathrm{last} -1}$, every register in $R_{\mathrm{last} -1}$ was covered by at least 3 processes,
every process in $R_{\mathrm{last}}$ is covered by at least one process.
So it only remains to bound $| R_{\mathrm{last}}| =j_{\mathrm{last}}$ from below. 

First we show that 
$|\idle{C_{\mathrm{last}}}|\leq 1$ is not possible.
Intuitively,
this is because during execution $(C_0; \gamma_1, \gamma_2, \ldots , \gamma_{\mathrm{last}})$
processes pause in such a way that each of the constructed configurations $C_1, \ldots , C_{\mathrm{last}}$
is $m$-constrained, which does not allow enough room to use $n-1$ processes.

To make this precise,  
let $\gamma$ denote $\gamma_1 \gamma_2 \ldots  \gamma_{\mathrm{last}}$, 
and say that process $p$ \emph{is associated with register $\reg$} 
if $\reg$ is the last register that $p$ covers during execution $(C_0; \gamma)$.
During the execution $(C_0; \gamma)$, 
processes no longer become associated with a register $\reg$ after $\reg $ becomes a member of ${R_i}$ for some $i$.
Let $f(\reg)$ be the smallest step number, $i$, such that $\reg \in R_i$ (and $f(\reg) = \mathrm{last} $ otherwise). 
Also, for each register $\reg$, let $g(\reg) = | \{ p ~|~  p \text{ is associated with } \reg \} |$. 
We must have $g(\reg) = \poised{C_{f{(\reg)}}}{\{\reg\}}$.
If $|\idle{C_{\mathrm{last}}}|\leq 1$, then each of $n-1$ processes is associated with a register.
So $ n-1  \leq  \sum_{\allRegs} g(\reg) = \sum_{\allRegs} \poised{C_{f{(\reg)}}}{\{\reg\}}$.
But by construction, each $C_i$ is $\ell_i$-constrained and therefore $m$-constrained.
Thus $\sum_{\allRegs} \poised{C_{f{(\reg)}}}{\{\reg\}}    \leq  \sum_{c=1}^{m}  (m -c) $.
But then $n-1 \leq  \sum_{c=1}^{m}  (m -c) = m(m-1)/2 $, which can hold only if $m  > \sqrt{2n} $
(since $n \geq 3$).

%

We can therefore conclude that the construction must have ended because
$\ell_{\mathrm{last}}-j_{\mathrm{last}}  \leq 2 $. 
So, now we show that if $\ell_{\mathrm{last}}-j_{\mathrm{last}}  \leq 2 $  then $j_{\mathrm{last}}$ is at least $m- \log n -2$. 
Let $\delta$ be the number of times that Case 2 occurred in the creation of
$(\gamma_1, C_1, j_1, \ell_1, R_1), \ldots ,
 (\gamma_{\mathrm{last}}, C_{\mathrm{last}}, j_{\mathrm{last}}, \ell_{\mathrm{last}}, R_{\mathrm{last}}) $.
Because $\ell_0 =m$ and $\ell_i$ decreases only for this case and only by one each time, $\ell_{\mathrm{last}} = m -\delta$.
Consider a step $i$  where Case 2 occurs, with $\gamma_i = \beta\sigma\beta'\sigma'$.  
By Lemma \ref{lemma:revisedCore},
$|\procs{\sigma}| \geq \lfloor |\idle{C_k}| / 2 \rfloor  $ so $|\procs{\sigma \sigma'}| \geq \lceil |\idle{C_k}| / 2 \rceil$.
Since $\idle{C_0} =n $ and  $\idle{C_i} < \idle{C_{i+1}}$ it follows that Case 2 can occur at most $\log n$ times.
Consequently, $\delta \leq \log n$ implying $\ell_{\mathrm{last}} \geq  m - \log n$.
Hence, 
$j_{\mathrm{last}} \geq \ell_{\mathrm{last}} -2 \geq  m - \log n - 2$.

This completes the proof of Theorem \ref{theorem:one-timeLbd}.

\section{A Simple One-Shot Timestamps Implementation Using \lowercase{\boldmath$\lceil n/2 \rceil$} Registers}

Algorithms~\ref{algo:simpleOneShotCompare} and \ref{algo:simpleOneShotGetTS} 
implement one-shot timestamps for $n$ processes using $\lceil n/2 \rceil$ registers 
and thus beat the space used by any register implementation of long-lived timestamps.
It is of interest only because of its simplicity; 
in Section \ref{section:one-shot-ub}, we improve on this space complexity with a more complicated algorithm, 
which shows that the space lower bound of Section \ref{section:one-shot-lwb} is asymptotically tight.  

The  \simpgetTS{} method by process $p$ reads each of the registers in sequence, updates the value of 
the register that $p$ shares by adding one to what $p$ read, and returns as $p$'s timestamp the
sum of all the values read. 
The \simpcompare{$t_1,t_2$} method returns the truth value of $t_1 < t_2$.

\begin{algorithm}[H]
\label{algo:simpleOneShotCompare}
  \caption{simple-compare($t_1,t_2$)}
  \Return{$t_1<t_2$}\;
\end{algorithm}

\begin{algorithm}[H]
\label{algo:simpleOneShotGetTS}
  \SetKwData{Sum}{sum}
  \caption{simple-getTS()}
  \tcp{$R[1\dots \lceil n/2\rceil ]$ is a shared array of multi-reader/2-writer registers 
    each with a value in $\{0,1,2 \}$ and initialized to $0$. Register $R[i]$ is written by processes $2i$ and $2i+1$.}
  \tcp{\Sum is a local variable}\;
  
  $\Sum:=0$ \;
  \For{ $i=1\dots\lceil n/2\rceil$}{
     \If{$i = \lceil p/2 \rceil $}{
       $R[i]:=R[i] +1$\;
     } 
     $\Sum:=\Sum + R[i]$\;
   }      
   \Return{$\Sum$}\;
\end{algorithm}

\begin{lemma}\label{lemma:one-shot-upper-bound} 
Algorithms~\ref{algo:simpleOneShotCompare} and \ref{algo:simpleOneShotGetTS} 
constitute a waitfree implementation of one-shot timestamps for an asynchronous system of $n$ processes. 
\end{lemma}

\begin{proof} 
Clearly both methods \simpcompare and \simpgetTS are waitfree.
Let $p$ and $q$ be two processors that perform a \simpgetTS method call 
and let $t_p$ and $t_q$ be their corresponding timestamps. 
Assume that $p$.\simpgetTS{} happens before $q$.\simpgetTS{}.  
Each process writes either 1 or 2 to its register and only writes 2 if it observed that its register 
already held 1. 
Because it is one-shot, any such observed 1, 
must have been written by the observing process' partner,
and thus the value in each register never decreases.
Consequently, the value of {\sf sum} also never decreases so $t_p \leq t_q$. 
Since $p.\simpgetTS{}$ happens before $q.\simpgetTS$, $q$'s {\sf sum} will also account for the additional 
$1$ that $q$ writes to its own register and that is not observed by $p$.
Therefore $t_p < t_q$. 
\end{proof}

\section{An Asymptotically Tight Space Upper Bound for One-Shot Timestamps}
\label{section:one-shot-ub}

We now present a waitfree algorithm for any timestamp system that invokes at most $M$ \getTS\ method calls,
which uses $\lceil 2 \sqrt{M} \rceil$ registers. 
In particular, the algorithm uses $ \lceil 2 \sqrt{n} \rceil$ registers for an $n$-process one-shot timestamp system,  
thus establishing Theorem~\ref{theorem:one-timeUbd} and showing that the space lower bound 
of Section~\ref{section:one-shot-lwb} is asymptotically tight.

Timestamps are ordered pairs $(rnd, turn)\in\mathbb{N}\times(\mathbb {N}\cup\{0\})$.  
The compare method simply compares timestamps lexicographically without accessing shared memory
(see Algorithm~\ref{algo:compare}).

\begin{algorithm}[H]
\label{algo:compare}
\caption{compare($(rnd_1,turn_1),(rnd_2,turn_2)$)}
\nl             \Return{$(rnd_{1} < rnd_{2}) \vee \bigl((rnd_{1} = rnd_{2}) 
\wedge (turn_{1} < turn_{2})\bigr)$}
\end{algorithm}

\subsection{The getTS algorithm}

Algorithm~\ref{algo:getTS} provides the \getTS method. 
It uses the parameter $m$, the number of shared registers, 
which is a function $m = f(M)$, where $M$ is the maximum number of \getTS method calls.
We will prove that $f(M)=\lceil 2\sqrt{M}\rceil$ suffices. 
Each process numbers its own \getTS{} method calls sequentially. 
The $k$-th time that $p$ invokes \getTS, it does so using ID $= p.k$.
We refer to these IDs as \emph{getTS-ids}.
When specialized to one-shot timestamps, 
ID can be just the invoking process' identifier. 


\begin{algorithm}[H]
\label{algo:getTS}
\caption{getTS(ID)}
\tcc{For the $k$-th invocation by process $p$, ID $= p.k$. }

\textbf{Shared:\\}        
        \Indp\Indp  
               $R[1\ldots m]$: array of multi-writer multi-reader registers, 
initialized to $\bot$\;
        \Indm\Indm
\textbf{Local:\\}
        \Indp\Indp  
               $r[1\ldots m]$ initialized to $\bot$\;
               $j$ initialized to $1$\;
               $myrnd$\;
        \Indm\Indm
\nl     \While{$R[j] \neq \bot$ }
        {       
\nl     $r[j] = R[j]$

\nl     $j = j + 1$
        }
\nl $myrnd = j - 1$     

\nl \For{$j = 1 \ldots myrnd - 1$}
{
\nl   \eIf{$R[myrnd + 1] == \bot$}{
\nl              \uIf{$r[myrnd].seq[j] == last(R[j].seq)$}{
\nl                             $R[j] =  \langle (\ID), myrnd \rangle$\;
\nl                             \Return{$(myrnd, j)$}
                                }
\nl                     \ElseIf{$R[j].rnd < myrnd$}                                {                
\nl                                     $R[j] = \langle (\ID), myrnd \rangle$\;                                                                             }
                        }{       
\nl                     \Return{$(myrnd + 1, 0)$}
                                }
}                           
\nl $r[1 \ldots m] =\scan{R[1],\dots,R[m]}$

\nl \If{$r[myrnd + 1] == \bot$}
{
\nl $R[myrnd + 1] = \langle (last(r[1].seq), \ldots , last(r[myrnd].seq), 
\ID),myrnd + 1 \rangle$ 
}
\nl                     \Return{$(myrnd + 1, 0)$}

\end{algorithm}

The shared data structure used in the \getTS{} method is an array of $m$ multi-writer multi-reader atomic registers.
The content of each register is either $\bot$ (the initial value) 
or an ordered pair $\langle seq, rnd \rangle$ where, $seq$ is a sequence of getTS-ids, and $rnd$ is a positive integer.
The algorithm maintains the invariant that for some integer $k \geq 0$ the first $k$ registers are non-$\bot$ and all other registers are $\bot$. 
Moreover, the sequence $R[j].seq$ for $j \leq k$ has length either $1$ or $j$.
The $i$-th element of $seq$ is denoted $seq[i]$, and $last(R[j].seq)$ is the last element of the sequence $R[j].seq$. 

The algorithm uses the well-known obstruction-free \scan method 
due to Afek, Attiya, Dolev, Gafni, Merritt and Shavit \cite{DBLP:journals/jacm/AfekADGMS93},
which returns a \emph{successful-double-collect}. 
A \emph{collect} reads each $R[1], \ldots, R[m]$ successively and returns the resulting \emph{view}.
A successful-double-collect($R[1], \ldots, R[m]$) repeatedly collects until two contiguous views are identical.
The \scan can be linearized at any point between its last two collects.  
Although this \scan is not wait-free in general, the use of it by Algorithm~\ref{algo:getTS} is. 
This is because, in any execution, 
each \getTS{} performs at most $m-1$ writes, so each \scan operation will be successful after a finite number of collects.
Since \scan is linearizable, we treat it as atomic for the remainder of this section. 

The idea of the algorithm is as follows.
An execution proceeds in phases. 
During phase $k$, $R[1]$ through $R[k-1]$ are non-$\bot$;
$R[k+1]$ to $R[m]$ are $\bot$; 
$R[k]$ is either written or some \getTS{} is poised to write to it for the first time. 
Every write to $R[k]$ during phase $k$ is a pair $\langle seq, rnd \rangle$,
which stores a sequence 
of $k$ getTS-ids in $seq$.
We say that register $R[i]$ is \emph{valid} if the phase is $k$ and 
the last entry stored in $R[i].seq$ equals the $i$-th entry 
stored in $R[k].seq$.

Roughly speaking, phase $k-1$ ends when some \getTS{q} method discovers that all registers $R[1]$ through  $R[k-1]$ are invalid. 
Then \getTS{q} performs a \scan, which returns the view $(r_1,\dots,r_{k-1},\bot,\dots,\bot)$. 
The $k$-th phase starts precisely at this scan. 
Then \getTS{q} prepares to write the sequence $(\ell_1,\dots,\ell_{k-1})$ to $R[k].seq$, where $\ell_i=last(r_i.seq)$ for $1\leq i\leq k-1$.
First imagine, for simplicity, that \getTS{q}'s scan and subsequent write to $R[k]$ occur in one atomic operation.
In that case, at the beginning of the $k$-th phase, every register $R[i]$, $1 \leq i \leq k-1$, is valid. 
Because \getTS{q} started phase $k$, it returns the timestamp $(k,0)$.

For the rest of phase $k$, any other \getTS{p} method that began in phase $k$ examines the registers $R[1]$ through $R[k-1]$ 
in this order looking for the first register that is valid.
If it finds one, say $R[i]$, it writes $\langle p,k \rangle$ to $R[i]$, 
thus \emph{invalidating} $R[i]$ by making $last(R[i].seq)$ differ from the $i$-th entry stored in $R[k].seq$,
and returns the timestamp $(k,i)$.
If it fails to find one, it will perform a scan and prepare to start phase $k+1$.
Observe that this algorithm works correctly if all \getTS{} calls are sequential: 
the \getTS{} that starts phase $k$ returns $(k,0)$ and the $j$-th \getTS{} call after that,
for $1\leq j \leq k-1$, invalidates $R[j]$ and returns $(k,j)$.

There are several complications and subtleties that arise due to concurrent \getTS{} executions. 
Suppose a \getTS{} that began in phase $k$ sleeps before it writes its invalidation to a register $R[i]$.
If it wakes up during some later phase $k'$, its write could invalidate $R[i]$ for phase $k'$ 
making timestamp $(k',i)$ unusable, and so increase the space requirements. 
Such damage is confined to at most one such wasted timestamp per \getTS{} method as follows.
Each \getTS{p} begins by setting its local variable, $myrnd_p$, to the largest value such that $R[myrnd_p]$ is non-$\bot$.  
Before each of its writes, \getTS{p} checks that $R[myrnd_p +1]$ is still non-$\bot$.  
If it is not, the phase must have advanced, and \getTS{p} can safely terminate with timestamp $(myrnd_p+1, 0)$.

A more serious potential problem due to concurrency occurs when
our simplifying assumption above (that the scan and subsequent write occur in one atomic operation) does not hold.
Suppose at the end of phase $k-1$, both \getTS{p} and \getTS{q} are poised to execute a \scan and then write the result into $R[k]$.
If, after \getTS{p}'s scan and before \getTS{q}'s scan, some ``old'' writes happen to some registers say $R[1],\dots,R[j]$,
the results of their {\scan}s will differ.
After both {\scan}s, \getTS{q}'s view matches all register values, but \getTS{p}'s view matches 
only the contents of $R[j+1]$,\dots, $R[k-1]$.
Now let both \getTS{p} and \getTS{q} proceed until both are poised to write the result computed from their view 
to $R[k]$, and suppose \getTS{p} writes first.
At this point, registers $R[1],\dots,R[j]$ are already invalid because of \getTS{p}'s out-of-date view. 
Another \getTS{a} starting at this point will invalidate $R[j+1]$ and return timestamp $(k,j+1)$.
If after that, \getTS{q} writes to $R[k]$, the first $j$ registers could become valid, 
and \getTS{b} beginning after \getTS{a} completes would invalidate $R[1]$ and return timestamp $(k,1)$, 
which is incorrect because it is less than \getTS{a}'s timestamp.
This problem is eliminated by ensuring that when \getTS{a} determines that a register $R[i]$ is invalid,
it will remain invalid for the duration of the phase. 
One way to achieve this is to have \getTS{a} overwrite the invalid register $R[i]$ with $\langle a, myrnd_a \rangle$ 
before it moves on to investigates the validity of $R[i+1]$.
This simple repair to correctness, however, can increase space complexity.  
Instead, the overwriting by \getTS{a} is done only when necessary. 
Specifically, \getTS{a} determines that a register $R[i]$ is invalid by reading a pair $\langle seq_i,rnd_i \rangle $ 
from $R[i]$ and finding that $last(seq_i)$ is not equal to its view of the $i$-th value in $R[k].seq$.
If $rnd_i =k$ then this invalidation cannot be due to an old write from an earlier phase, so no overwriting is needed. 
In the algorithm, \getTS{a} overwrites register $R[i]$ with $\langle a,k \rangle $ only when it read $rnd_i < k$.  

As we shall see, these additional techniques are enough to convert the idea of a timestamp object that is 
correct under sequential accesses to an algorithm for concurrent timestamps that is 
correct (Lemma \ref{lemma:correctnessOfOneTimeAlg}) and 
space efficient (Lemma \ref{lemma:spaceComplexityOfOneTimeAlg})
and  waitfree (Lemma \ref{lemma:waitfreedomOfOneTimeAlg}). 
These three lemmas, when specialized to the one-shot case,  
constitute the proof of Theorem \ref{theorem:one-timeUbd}.


\subsection{Algorithms~\ref{algo:compare} and \ref{algo:getTS} Correctly Implement Timestamps }
\label{subsec:oneTimeAlgCorrectness}

We isolate some properties of Algorithm~\ref{algo:getTS} that will serve to simplify 
both the correctness and complexity arguments. 
In the following, the local variable $x$ in the code of Algorithm~\ref{algo:getTS} is 
denoted by $x_{id}$ when it is used in the method call of \getTS{id}.

\begin{claim}
\label{claim:simpleProperties}
In any execution 
\begin{enumerate}
\item [(a)]
once the content of a shared register becomes non-$\bot$ it remains non-$\bot$; and
\item [(b)]
For any $j$,  $1 \leq j \leq m$, 
the value of last($R[j].seq$) that is written by each write to $R[j]$ is distinct. 
\end{enumerate}
In any configuration of an execution
\begin{enumerate}
\item [(c)] 
if any \getTS{id} has returned $(rnd, turn)$ then $R[rnd] \neq \bot$; and 
\item [(d)]
if $R[k] \neq \bot$ then $\forall k' \leq k$, $R[k'] \neq \bot$.
\end{enumerate}
\end{claim}

\begin{proof}
\begin{enumerate}
\item [(a)]
No \getTS{} method call ever writes $\bot$ to any shared register.
\item [(b)]
The only writes to a shared register occur at lines~8, 11 and 15. 
In any single instance of \getTS, say \getTS{id}, 
in each iteration $j$ of the for-loop (line~5), for $1 \leq j \leq  myrnd_{id} -1$,
at most one write occurs, either at line~8 or 11 but not both, and any such write is to $R[j]$.  
If \getTS{id} writes at line~15, it writes to $R[myrnd_{id}+1]$.
So, in any single execution of \getTS{id}, each register is written at most once.
Every write by \getTS{id} to a register $R[j]$ sets $last(R[j].seq) $ to $id$, 
which is distinct for each \getTS method call. 
\item [(c)]
\getTS{id} returned at line~9, 12 or 16.
We show that in all cases the register $R[rnd]$ was written before \getTS{id} returned. 
Then the claim follows by (a).
If \getTS{id} returned  in line~9 then $rnd=myrnd_{id}$,
and $R[myrnd_{id}] \neq \bot$ when the while-loop of \getTS{id} completes. 
If \getTS{id} returned in line~12, then $rnd = myrnd_{id} + 1$. 
Before returning, however, \getTS{id} evaluated the if-statement in line~6 to be false,
implying $R[myrnd_{id} + 1] \neq \bot$. 
If \getTS{id} returned in line~16, then $rnd = myrnd_{id} + 1$ 
and either \getTS{id} evaluated the if-statement in line~14 to be false, 
or \getTS{id} wrote to $R[myrnd_{id} + 1]$ in line~15.
In either case, $R[myrnd_{id} + 1]\neq\bot$ before \getTS{id} returned. 
\item [(d)]
Consider any write to a register $R[k]$ and suppose it occurs in the execution of \getTS{id}.
The while-loop of \getTS{id} confirms that all registers 
$R[1] $ through $R[myrnd_{id}]$ were previously non-$\bot$, before any write by \getTS{id}.
Writes only occur in lines~8, 11 and 15 of \getTS. 
Every write by \getTS{id} in lines~8 and 11 is to some register $R[j]$ where $j < myrnd_{id}$. 
A write in line~15 by \getTS{id} is to $R[myrnd_{id} + 1]$. 
So in all cases, when the write to $R[k]$ occurred, 
registers $R[1]$ through $R[k-1]$ were previously non-$\bot$. 
The claim follows by (a).
\end{enumerate}
\end{proof}

\begin{definition}
A \getTS{} method \emph{fails in iteration $j$ in line~6} if, in its $j$-th iteration of the for-loop (line~5),
the if-condition in line~6 returns false;
it \emph{fails in iteration $j$ in line~7} if, in its $j$-th iteration of the for-loop,
the if-condition in line~7 returns false; 
and it \emph{fails in iteration $j$}, if either it fails in iteration $j$ in line~6 or it 
fails in iteration $j$ in line~7.
\end{definition}

\begin{claim}\label{claim:mainCorrectness}
If $myrnd_p \geq  myrnd_q$ for two method calls \getTS{p} and \getTS{q},
and \getTS{p} writes to $R[j]$ before the $j$-th iteration of the for-Loop of \getTS{q} begins,
then \getTS{q} fails in iteration $j$. 
\end{claim}
\begin{proof}
$R[myrnd_p] \neq \bot$ when \getTS{p} executed line~1 of its while-loop for $j= myrnd_p$, 
and thus by Claim \ref{claim:simpleProperties}(a) remains non-$\bot$ after the while-loop completes. 

First, suppose $myrnd_{p} > myrnd_{q}$. 
By Claim \ref{claim:simpleProperties}(a) and (d), 
$R[myrnd_q+1] \neq \bot$ when \getTS{q} executes its $j$-th iteration of the for-loop.
So  the if-condition in line~6 of \getTS{q} returns false, and \getTS{q} fails at iteration $j$.

Now, suppose  $myrnd_{p} = myrnd_{q}$. 
\getTS{p} wrote to $R[j]$ after executing its while-loop and therefore after $R[myrnd_q]$ became non-$\bot$. 
The content of $r[{myrnd_{q}}]_q$, which $q$ read from $R[myrnd_q]$, 
came from the value of a scan taken during the execution of some \getTS when $R[myrnd_q] =\bot$.
Hence when \getTS{q} executes line~7 of the $j$-th iteration of the for-loop,
it compares the value of $r[{myrnd_{q}}]_q.seq[j]$, 
which is the value of $last(R[j].seq)$ that $R[j]$ had before $R[j]$ was written by \getTS{p}, 
to a value of $last(R[j].seq)$ that $R[j]$ had after this write.
So by Claim \ref{claim:simpleProperties}(b), 
this comparison must return false,  and \getTS{q} fails at iteration $j$.
\end{proof}

\begin{lemma}
\label{lemma:correctnessOfOneTimeAlg}
Provided $m = f(M) $ is large enough,
Algorithms~\ref{algo:compare} and \ref{algo:getTS} implement a timestamp object for any asynchronous shared memory
system of processes that invokes \getTS{} a total of at most $M$ times. 
\end{lemma}

\begin{proof}
Let  \getTS{p} and \getTS{q} return timestamps $(rnd_{p}, turn_{p} )$ 
and $( rnd_{q}, turn_{q} ) $ respectively. 
We need to show that if \getTS{p} happens before \getTS{q}, then 
compare($(rnd_p,turn_p),(rnd_q,turn_q)$) returns true.
That is, we need to show that $(rnd_{p} < rnd_{q}) \vee \bigl((rnd_{p} = rnd_{q})\wedge (turn_{p} < turn_{q})\bigr)$.

By Claim \ref{claim:simpleProperties}\,(c), after \getTS{p} completes, $R[rnd_{p}] \neq \bot$. 
Therefore by Claim \ref{claim:simpleProperties}\,(a) and (d), 
$R[1],\dots,R[rnd_p]\neq\bot$ throughout the method call \getTS{q}.
Thus, at line~4 after the while-loop,  \getTS{q} sets $myrnd_{q} \geq rnd_{p}$. 
If \getTS{q} returns at line~12 or 16, 
$rnd_{q} = myrnd_{q} + 1 \geq rnd_{p} + 1$ implying $(rnd_{p} < rnd_{q})$ so 
compare($(rnd_p,turn_p),(rnd_q,turn_q)$ returns true as required.
If  \getTS{q} returns at line~9, 
and $myrnd_{q} > rnd_{p}$,  
then again compare($(rnd_p,turn_p),(rnd_q,turn_q)$ returns true.
Therefore, suppose that \getTS{q} returns at line~9 and  $rnd_q = myrnd_q = rnd_p$.
In this case, $turn_{q}$ is non-zero. 
If $p$ returns at line~12 or 16, $turn_{p}$ is zero. 
So again compare($(rnd_p,turn_p),(rnd_q,turn_q)$ returns true.

The only remaining case is that both \getTS{p} and \getTS{q} return at line~9 and $rnd_q= myrnd_q = rnd_p=myrnd_p$.
In this case, we show that  $turn_q > turn_p$ by proving that \getTS{q} fails at every iteration $1$ through $turn_p$. 
By Lemma \ref{claim:mainCorrectness}, it suffices to show that for every $j$, $1 \leq j \leq turn_p$,
there is some \getTS{$p'$}, with  $myrnd_{p'} \geq myrnd_{q}$ that writes to $R[j] $ before \getTS{q} begins iteration $j$.
For \getTS{p}, the if-condition in line~7 must have returned false 
for all iterations $1$ through $turn_p - 1$, and then returned true in iterations $turn_p$. 
For $j < turn_{p}$, when \getTS{p} fails at iteration $j$, it reads $R[j]$ (line~10).
If this read shows $R[j].rnd \geq myrnd_p$ 
there must be a \getTS{$p'$}, with $myrnd_p' \geq myrnd_p$ that wrote this.
If the read shows $R[j].rnd < myrnd_p$, 
then \getTS{p} itself writes to $R[j]$ changing $R[j].rnd$ to $myrnd_p$ (line~11).
For $j = turn_{p}$ process $p$ itself writes into register $R[j]$. 
In all cases, the write happened before \getTS{q}.
\end{proof}

\subsection{Space Complexity of Algorithm~\ref{algo:getTS} }
\label{subsec:oneTimeAlgComplexity}

Fix an arbitrary execution $E$ that contains at most $M$ \getTS{} invocations.
In this subsection we prove no register beyond $R\big[\left\lceil 2\sqrt{M}\right\rceil\big]$
is accessed in $E$.

The proof proceeds as follows.
We partition $E$ into \emph{phases}. 
Phase $0$ starts at the beginning of $E$.
Phase $\varphi \geq 1$ starts at the point of the first \scan (line~13) by some \getTS{p}, for which $myrnd_{p} = \varphi - 1$. 
We say that phase $\varphi$ \emph{completes during $E$}, if phase $\varphi + 1$ starts during $E$.
Call the first write to $R[j]$ during any phase an \emph{invalidation write}.
First, Claim \ref{claim:registersNotWrittenPerPhase} shows that 
only registers $R[1]$ through $R[\varphi]$ can be written during phase $\varphi$.
Next, Claim \ref{claim:invalidationsPerPhase} 
establishes that if phase $\varphi$ completes then it contains exactly $\varphi$ invalidation writes. 
Finally, we define a charging mechanism that charges each invalidation write in $E$ to 
some write in $E$ in such a way that there are at most $2$ charges to all the writes of any one \getTS{} invocation.
This gives us Claim \ref{claim:totalInvalidations}, 
which states that there are a total of at most $2M$ invalidation writes.

Therefore, the total number of phases, $\Phi$, satisfies: $\sum^{\Phi}_{\varphi = 1} \varphi \leq 2 M$.
Hence, $\Phi< 2\cdot \sqrt{M}$. 
The algorithm uses a final sentinel register that is read but never written, and  always contains the initial value $\bot$.
So at most $\lceil 2\cdot \sqrt{M} \rceil$ registers are accessed in $E$.
Therefore, once Claims \ref{claim:registersNotWrittenPerPhase}, 
\ref{claim:invalidationsPerPhase} and \ref{claim:totalInvalidations} 
are proved (below) we have the following:
\begin{lemma}
\label{lemma:spaceComplexityOfOneTimeAlg}
Algorithm \ref{algo:getTS} uses at most $\lceil 2 \sqrt{M}\rceil$ registers for $M$ \getTS{} operations. 
\end{lemma}

\subsection*{Technical claims}

The proof of Lemma \ref{lemma:spaceComplexityOfOneTimeAlg},
via Claims \ref{claim:registersNotWrittenPerPhase}, \ref{claim:invalidationsPerPhase} and \ref{claim:totalInvalidations},
is the most challenging part of our arguments concerning Algorithm \ref{algo:getTS}. 
First, we encapsulate the relationship between 
the value of the variable $myrnd_p$ of a \getTS{p} method call 
and the phase number $\varphi$ during which \getTS{p} writes to certain registers.

\begin{claim}\label{claim:phaseRound}\samepage
\ \begin{itemize}
	\item [(a)] If \getTS{p} writes to register $R[i]$ when $R[i+1]=\bot$, then that write occurs in line~15.
	\item[(b)] \getTS{p} executes line~15 during some phase $\varphi \geq myrnd_{p} + 1$.
	\item[(c)] \getTS{p} executes line~4 during some phase $\varphi' \geq myrnd_{p}$.
\end{itemize}
\end{claim}
\begin{proof}
(a) If $w$ is a  write by \getTS{p} to $R[i]$ in line~8 or 11, then $i\leq myrnd_p-1$.
When \getTS{p} previously read $R[myrnd_p]$ in line~2, its value was non-$\bot$, 
so, by Claim~\ref{claim:simpleProperties}(a) and (d) $R[i+1]$ is non-$\bot$ when $w$ occurred.
Hence, any write to  $R[i]$ when $R[i+1]=\bot$ could not have occurred at line line~8 or 11, 
and thus could only occur at line~15.

(b) When \getTS{p} executes line~15, it has already executed its \scan in line~13.
By definition of phase, 
if phase $myrnd_p+1$ had not already begun before this \scan occurred, 
then it began with this scan. 

(c) Before \getTS{p} executes line~4, its while-loop terminated because $R[myrnd_{p}] \neq \bot$ and $R[myrnd_p +1] =\bot$. 
By (a), $R[myrnd_p]$ must have previously changed from $\bot$ to non-$\bot$, when some \getTS{r} executed line~15.
When \getTS{r} executes this write, it wrote to $R[myrnd_r +1]$, so $myrnd_r +1 = myrnd_p$.
Before this write, \getTS{r} executed a \scan at line 13, which either started phase $myrnd_r +1$, or 
phase $myrnd_r+1$ had already started. 
Thus, $myrnd_r+1 =myrnd_p$ started before \getTS{p} executed line~4.
\end{proof}

\begin{claim}\label{claim:failingImpliesPriorWrites}
If during phase $myrnd_{p}$, \getTS{p} fails iteration $i$ at line~7, 
then during phase $myrnd_{p}$ and before the failure, there was a write to $R[i]$ and a write to $R[myrnd_{p}]$. 
\end{claim}
\begin{proof}
Let $v=(v_1,\dots,v_k)$ be the value of the sequence stored in $R[myrnd_{p}].seq$ when \getTS{p} reads that register in line~2. 
Let \getTS{b} be the method call that wrote $v$ to $R[myrnd_{p}].seq$. 
The  while-loop of \getTS{p} terminated 
when \getTS{p} read $R[myrnd_p +1] =\bot$ in the $(myrnd_p +1)$-th iteration of its while-loop 
after reading $R[myrnd_{p}] \neq \bot$ in its previous iteration.
By Claim \ref{claim:simpleProperties}\,(a), $R[myrnd_{p} + 1] = \bot$ when \getTS{b} wrote $v$ to $R[myrnd_{p}]$. 
Therefore, by Claim \ref{claim:phaseRound}\,(a), \getTS{b} wrote to $R[myrnd_{p}]$ in line~15 and thus $myrnd_{b} = myrnd_{p} - 1$. 
By Claim~\ref{claim:phaseRound}\,(b), \getTS{b}'s write to $R[myrnd_{p}]$ happens during phase $myrnd_{p}$ or a later phase. 
Because this write happens before \getTS{p} reads $R[i]$ in line~7 during phase $myrnd_{p}$, 
\getTS{b}'s write to $R[myrnd_{p}]$ occurs during phase $myrnd_{p}$ before \getTS{p} fails at iteration $i$. 

Before executing line~15, \getTS{b} executed a \scan in line~13, 
and obtained $v_i$ for the value of $last(R[i].seq)$.
By the assumption that \getTS{p} fails at iteration $i$ in line~7, 
$v_i\neq last(R[i].seq)$ when \getTS{p} reads $R[i]$ in line~7. 
Therefore, there must have been a write to $R[i]$ between \getTS{b}'s scan and \getTS{p}'s read at line 7. 
This write must have occurred in phase $myrnd_{p}$ because,  by the definition of phase, 
either phase $myrnd_p$ began before \getTS{b}'s scan or it began at this scan.
Thus, a write to $R[i]$ occurs in phase $myrnd_p$ before \getTS{p} fails iteration $i$.
\end{proof}

\subsection*{Counting invalidation writes per completed phase}

Our goal is to show that during phase $\varphi$ there is exactly one invalidation write to each register
$R[1]$ through $R[\varphi]$, and no other registers are written.

\begin{claim}
\label{claim:registersNotWrittenPerPhase}
Only registers $R[1], \ldots, R[\varphi]$ are written during phase $\varphi$.
\end{claim}
\begin{proof}
Until some \getTS{} has executed line~15, and thus phase $1$ has started, no register is written. 
Hence, the claim is trivially true for $\varphi = 0$. 
Now let $\varphi \geq 1$.
If \getTS{q} writes to $R[j]$ in lines~8 or 11, 
then by Claim~\ref{claim:phaseRound}\,(c), $\varphi\geq myrnd_q$, and by the semantics of the for-loop, $j<myrnd_q$.
If $q$ writes to $R[j]$ in line~15, then $j=myrnd_q+1$ and by Claim~\ref{claim:phaseRound}\,(b), $\varphi \geq myrnd_{q} + 1$.
Hence, in either case $j \leq \varphi$.
\end{proof}

\begin{claim}
\label{claim:registersWrittenPerPhase}
If phase $\varphi$ completes in $E$, 
then for each $1\leq j \leq \varphi$, there is at least one write to $R[j]$ during phase $\varphi$.
\end{claim}
\begin{proof}
By definition, phase $\varphi +1  \geq 1$ begins at the first \scan (line~13) by some \getTS{p}, 
for which $myrnd_{p} = \varphi$. 
Since \getTS{p} executes line~13, its  call does not return during the for-loop. 
Therefore, this scan can happen only if this \getTS{p} fails in iteration $j$ at line~7 for every $1 \leq j \leq \varphi-1$.
Thus, by Claim~\ref{claim:failingImpliesPriorWrites}, 
a write to register $R[\varphi]$ and a write to $R[j]$ for every $1 \leq j \leq \varphi-1$
happens in phase $\varphi$. 
\end{proof}

\begin{claim}
\label{claim:invalidationsPerPhase} 
There are exactly $\varphi$ invalidation writes in any completed phase $\varphi$. 
\end{claim}

\begin{proof}
By Claims \ref{claim:registersNotWrittenPerPhase} and \ref{claim:registersWrittenPerPhase},
exactly the registers $R[1]$ through $R[\varphi]$  are written during phase $\varphi$.
The set of first writes to each of these registers during phase $\varphi$ is, 
by definition, the set of invalidation writes in phase $\varphi$.
\end{proof}

\subsection*{Counting all invalidation writes} 

We rely on some definitions and factor out some sub-claims 
to facilitate the proof of Claim \ref{claim:totalInvalidations} below. 

\begin{claim}
\label{claim:notInvalidationWrite}  
If a write at line~11 by \getTS{p} happens during phase $myrnd_p$, then that write is not an invalidation write. 
\end{claim}
\begin{proof}
Let $w$ be a write at line~11 by \getTS{p} to register $R[j]$ that occurs during phase $myrnd_p$. 
Then $w$ happens only if \getTS{P} fails  at iteration $j$ at line~7. 
So, according to Claim \ref{claim:failingImpliesPriorWrites}, there was a previous write to $R[j]$ during phase $myrnd_p$.
Hence, $w$ is not an invalidation write.
\end{proof}

There can be an interval between when the first \getTS{q} with $myrnd_q = \varphi-1$ does a scan at line~13 
thus starting phase $\varphi$,
and  when the first write to $R[\varphi]$ happens, 
which is the first point at which other \getTS{} method calls can discern that the current phase is (at least) $\varphi$.
To capture this, say that the phase $\varphi$ is 
\emph{invisible} from the beginning of phase $\varphi$ to the step before the first write to $R[\varphi]$ and 
\emph{visible} from the first write to $R[\varphi]$ to the end of phase $\varphi$.

\begin{claim}
\label{claim:lastWrite}
Any write by \getTS{p} at line~8 or at line~15 
or any write at line~11 that happens after the phase $myrnd_p +1$ becomes visible, 
is \getTS{p}'s last write. 
\end{claim}
\begin{proof}
Algorithm \ref{algo:getTS} returns after any write at line~8 or line~15, so such a write is the last write of the method call.
Now consider a line~11 write $w$ by \getTS{p}.
If $w$ happens anytime after phase $myrnd_p+1$ becomes visible, 
then after $w$, \getTS{p} will discern that the phase number has increased 
when it reads $R[myrnd_p+1]$ to be non-$\bot$ 
either at line~$6$ in the next iteration of the for-loop or, if the for-loop is complete, at line~14. 
In either case \getTS{p} returns without another write, so such a write is also the last write of the method call.
\end{proof}

\begin{claim}
\label{claim:totalInvalidations}  
There are at most $2M$ invalidation writes in execution $E$.
\end{claim}
\begin{proof} 
Define:
\begin{eqnarray*}
A &=&  \{ w ~|~ \text{$w$ is a first invalidation write by some \getTS{} method} \} \\ 
B &=&  \{ w ~|~ \text{$w$ is the last write by some \getTS{} method and $w$ is an invalidation write} \}  \\ 
C &=&  \{ w ~|~ \text{$w$ is the last write by some \getTS{} method and $w$ is not an invalidation write} \}
\end{eqnarray*}
Let $ W^*$ be the disjoint union of $A, B$ and $C$. 
Since each \getTS{} can have at most one write in $A$ and at most one write in either $B$ or $C$ but not both
it follows that $| W^*| \leq 2M$.
Let $W$ be the set of all writes, and let $I$ be the set of all invalidation writes during execution $E$.
We will define a function $f: I \rightarrow W $.
Then, it will suffice to show that $f$ is injective with range a subset of $W^*$.

Define:
\begin{eqnarray*}
I_1 &=& \{ w ~|~ w \text{ is an invalidation write at line~8 or 15} \}  \\
I_2 &=& \{ w ~|~ \exists \ \getTS{p} \text{ satisfying ($w$ is an invalidation write at line~11 by \getTS{p}) and } \\
    & &     \text{($w$ happens after the beginning of the visible part of phase $myrnd_p +1$) } \} \\
I_3 &=& \{ w ~|~ \exists \ \getTS{p} \text{ satisfying ($w$ is an invalidation write at line~11 by \getTS{p}) and } \\ 
    & &     \text{($w$ happens during the invisible part of phase $myrnd_p +1$ ) and } \\
    & &     \text{($w$ is \getTS{p}'s first invalidation write)} \} \\
I_4 &=& \{ w ~|~ \exists \ \getTS{p} \text{ satisfying ($w$ is an invalidation write at line~11 by \getTS{p}) and } \\
    & &      \text{($w$ happens during the invisible part of phase $myrnd_p +1$ ) and }  \\
    & &     \text{($w$ is not \getTS{p}'s first invalidation write)} \} 
\end{eqnarray*}
Let $w$ be a write by \getTS{p}.
Then $w$ happens at line~8, or line~11 or line~15,  
and, by Claim \ref{claim:phaseRound}\,(c),
$w$ happens in some phase $\varphi$ satisfying $\varphi \geq myrnd_p$.
By Claim \ref{claim:notInvalidationWrite}, 
if $w$ is a write at line~11 that occurs during phase $myrnd_p$, 
then $w$ is not an invalidation write. Hence $I_1 \cup I_2 \cup I_3 \cup I_4 = I$.
Also, clearly $ I_1, I_2 , I_3$ and  $I_4$ are mutually disjoint.  
Therefore $ \{ I_1 , I_2,  I_3, I_4 \} $ is a partition of $I$. 

For all $ w \in I_1 \cup I_2 \cup I_3$ define $f(w) =w$.
By definition,  $I_3 \subseteq A$, and 
by Claim~\ref{claim:lastWrite}, $I_1 \cup I_2 \subseteq B$.
So, $f$ maps $I_1 \cup I_2 \cup I_3$ to the disjoint union of $A$ and $B$ and
clearly, $f$ is injective on $I_1 \cup I_2 \cup I_3$.


It remains to map $I_4$ to $C$ and show this map is injective.
Let $w$ be a write in $I_4$ by \getTS{p} to register $R[i]$. 
By definition of $I_4$, $w$ occurs during the invisible part of phase $myrnd_p +1$, 
and there is another invalidation write, say $u$, by \getTS{p} that precedes $w$ in $E$. 
Claims \ref{claim:phaseRound}\,(c), \ref{claim:notInvalidationWrite} and \ref{claim:lastWrite} imply that
$u$ is a line~11 write that also occurs during the invisible part of phase $myrnd_p +1$.

In line~10, before executing $w$, \getTS{p} reads a value $x < myrnd_{p}$ from $R[i].rnd$. 
Let $o$ be this read operation. 
Operation $o$ occurs after $u$ and before $w$ and thus also during the invisible part of phase $myrnd_{p} + 1$. 
Define $f(w)$ to be the write operation that wrote the value to $R[i]$ that was read by $o$.

We now show that $f(w)$ is in $C$. 
Let \getTS{a} be the method call that starts phase $myrnd_{p} + 1$ by executing a \scan in line~13. 
Then $myrnd_{a} = myrnd_{p}$, and during phase $myrnd_p$,
\getTS{a} fails at iteration $j$ at line~7 and thus executes line~10, for all $j = 1, \ldots, myrnd_{p} - 1$.
In particular for $j = i$, 
\getTS{a} either writes the value $myrnd_{a} = myrnd_{p} > x$ to $R[i].rnd$ in line~11, or in line~10 it reads $R[i].rnd \geq myrnd_{a} > x$. 
Hence, $f(w)$, which writes the value $x$ to $R[i].rnd$ that is read by $o$, must happen after the $i$-th iteration of \getTS{a}'s for-loop and before $o$. 
Furthermore, $f(w)$ cannot happen during phase $myrnd_p +1$, because otherwise $w$ would not be an invalidation write. 
We conclude that $f(w)$ is a write to $R[i]$ that happened in phase $myrnd_a$ after \getTS{a} failed at iteration $i$,
and hence, by Claim \ref{claim:failingImpliesPriorWrites}, $f(w)$ is not an invalidation write.

Let \getTS{b}  be the method call that executes $f(w)$. 
When \getTS{a} finished its while-loop, $R[myrnd_{a}] = R[myrnd_{p}] \neq \bot$. 
Hence, by Claim~\ref{claim:simpleProperties}\,(a), $R[myrnd_{p}] \neq \bot$ when $f(w)$ occurs. 
Since $myrnd_b = x < myrnd_p$, by Claim~\ref{claim:simpleProperties}\,(a) and (d),
$R[myrnd_{b} + 1] \neq \bot$ when \getTS{b} executes either the if-statement in line~6 (in iteration $i+1$ of its for-loop)
or in line~14 (if its for-loop is completed because $i = myrnd_b-1$). 
In either case, \getTS{b} returns without another write.
Therefore, $f(w)$ is the last write by \getTS{b} and is not an invalidation write, so $f(w) \in C$ by definition of $C$.
Finally, we show that $f(\cdot)$ on the domain $I_4$ is injective.
If $f(w)$ occurs in phase $\varphi$ then, 
as we have just seen, $w$ occurs during the invisible part of phase $\varphi +1$.
Suppose $f(w) = f(w')  $ where $w,w' \in I_4$.
Then $w$ and $w'$ are both invalidation writes to the same register during the same phase $\varphi+1$. 
But this is impossible since there can be only one invalidation write per register per phase.
\end{proof}

\subsection{Algorithms~\ref{algo:compare} and \ref{algo:getTS} are Waitfree }

\begin{lemma}
\label{lemma:waitfreedomOfOneTimeAlg}
Algorithms~\ref{algo:compare} and \ref{algo:getTS} are wait-free provided the bound $M$ on the number of allowed \getTS method call is fixed in advance.
\end{lemma}
\begin{proof}
Clearly, \compare is wait-free.
The number of registers $m$ provided for \getTS{} is at least one more than the number that can be written, 
so the last register $R[m]$ is always $\bot$.
Since each iteration of the  while-loop increments $j$ until $R[j]=\bot$ is read, 
the while-loop can iterate at most $m-1$ times. 
Similarly, since $myrnd$ is the index of a non-$\bot$ register,  
the for-loop can iterate at most $m-2$ times. 
All, operations inside and outside the while and for loops, except the \scan, are wait-free. 
Hence, it remains to show that all calls of \scan terminate within a bounded number of steps.
It is immediate from the code that each \getTS{} writes to each register at most once, 
implying each \getTS{} method writes fewer than $m = \lceil 2\sqrt{M} \rceil$ times.
Thus, after a finite number of reads during the \scan, 
the scanning process must see no more changes to registers, 
and so will  achieve a successful double collect and terminate.
\end{proof}

\section{Further remarks}

The lower and upper bounds for long-lived and one-shot timestamps compare and contrast in several ways.  
In the execution constructed in the lower bound for one-shot timestamps, 
each process that participates in a block write, takes no further steps in the computation.  
As a consequence, the proof actually applies without change if each register is replaced by any historyless object. 
The asymptotically matching upper bound is, however, achieved using registers.
In contrast, 
our proof of the lower bound for long-lived timestamps does not extend to historyless objects. 
So it remains an open question whether there is an implementation 
of long-lived timestamps from a sub-linear number of historyless objects. 
Both the long-lived and the one-shot lower bounds apply even to non-deterministic solo-terminating algorithms, 
while the asymptotically matching algorithms are wait-free.

The upper bound for one-shot timestamps applies for any bounded number of \getTS{} method invocations.
The covering argument in the proof of the lower bound, however, prevents any similar generalization: it 
depends on each process invoking at most one \getTS{}.

The one-shot algorithm generalizes even to the situation where the number of \getTS{} method invocations 
is not bounded, provided that the system could acquire additional registers as needed. 
In this case however, progress would be non-blocking only instead of wait-free.

\begin{acks}
The authors thank Faith Ellen for valuable comments on an earlier draft of parts of this paper. 
\end{acks}

\bibliographystyle{acmsmall}
\bibliography{timestamps}


\end{document}